\documentclass[letterpaper, onecolumn,11pt]{IEEEtran}

\usepackage{amsmath}
\usepackage{mathtools}
\usepackage{amsmath}
\usepackage{amssymb}

\DeclareFixedFont{\auacc}{OT1}{phv}{b}{it}{18}   
\DeclareFixedFont{\newauacc}{OT1}{ptm}{b}{rm}{12}   
\usepackage{hyperref}

\usepackage{lipsum}

\usepackage{cite}
\usepackage{array}
\usepackage{amsthm}
\usepackage{amsmath}
\usepackage{tabulary}
\usepackage{multirow}
\usepackage{subfig}
\usepackage{graphicx}
\usepackage[font=small]{caption}
\usepackage{etoolbox}
\usepackage{color}

\usepackage{pifont}
\usepackage{pgfplots}
\usepackage{bbm}
\usepackage{paralist}
\usepackage{url}
\usepackage{amssymb}
\usepackage[]{algorithm2e}
\usepackage{algpseudocode}
\usepackage{algorithmicx}
\usepackage{dsfont}

\graphicspath{ {./Plots_and_Figures/} }

\newtheorem{theorem}{Theorem}

\newtheorem{corollary}{Corollary}

\newtheorem{lemma}{Lemma}

\newtheorem{definition}{Definition}

\newtheorem{conjecture}{Conjecture}
\newtheorem{proposition}[theorem]{Proposition}
\newtheorem{claim}{Claim}

\begin{document}

\title{Spatial Birth-Death Wireless Networks
}

%
%

\author{Abishek Sankararaman and Fran\c{c}ois Baccelli \thanks{ \noindent A. Sankararaman is with the Dept of Electrical and Computer Engineering (ECE), UT Austin (Email: abishek@utexas.edu).} \thanks{  F. Baccelli is with the Dept of Mathematics and ECE, UT Austin (Email: francois.baccelli@austin.utexas.edu).} \thanks{An abstract of this paper was presented at Allerton Conference on Communication, Control and Computing, 2016.} }
\date{}
\maketitle


%
%
%


\maketitle

\begin{abstract}

We propose and study a novel continuous space-time model for wireless networks which takes into account the stochastic interactions in both space through interference and in time due to randomness in traffic. Our model consists of an interacting particle birth-death dynamics incorporating information-theoretic spectrum-sharing. Roughly speaking, particles (or more generally wireless links) arrive  according to a Poisson point process on space-time, and stay for a duration governed by the local configuration of points present  and then exit the network after completion of a file transfer. We analyze this particle dynamics to derive an explicit condition for time ergodicity (i.e. stability) which is tight. We also prove that when the dynamics is ergodic, the steady-state point process of links (or particles) exhibits a form statistical \emph{clustering}. Based on the clustering, we propose a conjecture which we leverage to derive approximations, bounds and asymptotics on performance characteristics such as delay and mean number of links per unit-space in the stationary regime. The mathematical analysis is combined with discrete event simulation to study the performance of this type of networks.

\end{abstract}



\section{Introduction}

We consider the problem of studying the spatial dynamics of Device-to-Device (D2D) or ad-hoc wireless networks. Such wireless  networks have received a tremendous amount of attention, due on the one hand to their increasing ubiquity in modern technology and on the other hand to the mathematical challenges in their modeling and performance assessment. Wireless is a broadcast medium and hence the nodes sharing a common spectrum in space interact through the interference they cause to one another. Understanding the limitations due to interference and theoretically optimal protocols in such a static spatial setting has long been considered in network information theory under the interference channel \cite{shannon1961}. The full characterization of the interference channel is however a long standing open-problem in network information theory.
\\

 In recent years,  Stochastic Geometry (\cite{haenggi_book}, \cite{Baccelli_sg_book}) has emerged as a way of assessing performance of wireless links in  large-scale networks interacting through interference in space. These tools have been very popular to model and analyze wireless system performance for a variety of network architectures including D2D networks, mobile-ad hoc networks \cite{baccelli_aloha} and cellular networks \cite{hetnet}. However, the main drawback in these models is that they do not have a notion of temporal interaction and do not allow one to represent random traffic (they usually rely on a ``full-buffer'' assumption, i.e., every link always has a packet to transmit). 
 \\


 This additional dimension of interaction among wireless links sharing a common spectrum adds to the complexity of their performance analysis but nonetheless is very crucial to understand network performance. Most prior work aiming at studying the temporal interaction of links model spatial interactions through binary on-off behavior encoded by interference or conflict graphs. The temporal interactions are then modeled using queuing theoretic ideas of flow based models (for ex: \cite{Borst_Flows}, \cite{Flow_Srikant}, \cite{Flow_Shah}). Such flow models have a long history in applied mathematics and engineering. They were initially proposed to study dynamic resource allocation in wired networks (\cite{Kelly_Flow}, \cite{BW_Sharing}), and were subsequently used to model and study wireless networks. Flow based queuing models have inspired many seminal results in applied probability and networks in the past. The main drawback in employing such models in a wireless scenario however is that the spatial and information-theoretic interactions are overly simplified and not captured precisely.
\\

Motivated by this, we propose a new spatial flow model, which uses the continuum space to model  link interaction through interference as prescribed by the information-theo\-re\-tic setting, and also takes into account the interaction of links across time due to traffic variations. Roughly speaking, our model consists of an interacting particle system in space, where links which is a transmitter-receiver pair arrive in space according to a Poisson Point Process in space-time. The transmitter of each link has a file which it wants to transmit to its corresponding receiver. A link exits the network upon completion of this file transfer. The instantaneous rate at which a transmitter can transmit a file to its receiver is given by the instantaneous Shannon rate, which in turn depends on the geometry of the other transmitters in the network transmitting at that instant to their respective receivers. We study this space-time dynamics to identify a phase-transition in the arrival rate such that each link can be guaranteed to exit in finite time almost surely. The model and the question of phase-transition is formalized in Section \ref{sec:system_model}. To the best of our knowledge, the analysis of such continuum space-time models for wireless networks has not been considered so far.
\\

 The mathematical framework we follow for spatial birth-death processes  has been studied in different contexts in the probability literature starting from the work of Preston \cite{Preston}. In recent years, \cite{Garcia} and \cite{Penrose} have also studied in great detail, the problem of general spatial birth and death process which is the basis of our modeling. From a methodological point of view, the work of \cite{baccelli_p2p} is the closest in spirit to our work as it also studies a space-time interacting particle process (of a wireline peer-to-peer network). There are several fundamental differences between the model of \cite{baccelli_p2p}, which is intrinsically stable, and exhibits repulsion, and our model, which is potentially unstable and which exhibits attraction (clustering). Another difference from \cite{baccelli_p2p} is that the death-rate (defined later) is a linear-function of the state whereas  our model  is  non-linear because  of the information-theoretic formulation, thereby making the analysis more challenging. Nevertheless, we use  some of the ideas developed in that paper.
\\

From an information-theoretic viewpoint, one can interpret our model and the phase-transition result as a form of dynamic network capacity. Our network model can be interpreted as consisting of arrivals of a single antenna Gaussian additive noise point-to-point channels in space. At each instant of time, the network is a random realization of an interference network operating under the scheme of treating interference as noise. The point-to-point channels exit the network upon completion of a file transfer i.e. with the departures happening in a space-time correlated way determined by our dynamics which in turn is derived from the capacity region of an interference channel under treating interference as noise. The phase-transition results in Theorems \ref{thm:necessary_condn} and \ref{thm:sufficient_condn} give the maximum rate of arrival that can be supported in the network under the scheme of treating interference as noise. Our model and the framework could potentially be generalized to  consider the dynamic capacity of other channels like the Multiple Access Channels or Broadcast channels instead of the Gaussian point-to-point channel considered in this paper. In these models, each arrival could consist of a single transmitter and multiple receivers or multiple transmitters and a single receiver which  form a basic unit of the network. This network can then be modeled to evolve in time through dynamics similar in spirit to Equation (\ref{eqn:dynamics_defn}). It is beyond the scope of the present paper however to pose the problem precisely in the case of multiple access or broadcast channels to derive a phase-transition for dynamic capacity. However, in Section \ref{sec:mimo_extnsions}, we present the extension of our model to the case of point-to-point channel where the transmitters and receivers have multiple antennas, i.e. the point-to-point Multiple Input Multiple Output (MIMO) channel. We then analyze and study a special case of the MIMO dynamics in Section \ref{sec:mimo_special} which can be derived as a corollary of the single antenna analysis. 
\\

%
%

Our model also presents a new form of  single server queuing network. Based on our model description in Section \ref{sec:system_model}, one can  come up with two natural queuing model bounds to study the performance of our model. One can construct a `worse' system by assuming that there is no distance dependent attenuation and all transmitters contribute the same interference to any receiver. This system will predict larger delays than our original system since the interference is higher. Moreover, since there is no geometry, this upper bound system is equivalent to an $M/M/1$ generalized processor sharing system. On the other hand, to come up with lower bounds for delay, one can totally neglect interference and assume that the different links do not interact at all. This assumption
will render our model equivalent to an $M/M/\infty$ system. One of our main messages in the paper is that  simplifying our model to any of the above two dynamics which neglects spatial structure to provide bounds on delay leads to estimates for delay  which are very poor (as demonstrated in Section $V.E$). Thus, we really need to consider the spatial structure as done in Section $IV$ to come up with estimates for delay and performance. The evolution of our model thus presents a novel behavior of stochastic dynamics that cannot be captured by a queuing model that neglects spatial interactions.
\\

From an engineering viewpoint, this work is motivated by emerging interest in applications like Device-to-Device (D2D) networks and Internet of Things (IoT). These two applications can be viewed as an instance of our abstract mathematical model which is more general. D2D is being considered as a viable networking architecture in future cellular standards to improve system capacity by offloading some traffic from base-station to other mobile devices that have the same content. Some of the more important use cases for such offloading  are in a crowded setting (like a stadium or a concert) where there is a huge density of mobile devices. Another important application of D2D is in enabling cellular operators to provide ``proximity based services''. In such settings, a mobile may access content (which we model as files) from nearby mobile users possessing the content (which may be likely owing to geographical and temporal proximity) rather than from a base-station. Such networking architectures are being envisioned to both reduce the load on the base-stations and also to develop new markets for mobile services.  Thus, a snapshot of a D2D network will  resemble our model with some mobile devices connecting to and downloading files from other mobile devices that are nearby. IoT is another technology gaining  momentum due to the vast market opportunities to develop user applications that leverage the IoT network (for instance in tracking sensors for health, security etc).  This network  also resembles a wireless ad-hoc network with different things communicating occasionally data to each other or to a central access point using the shared wireless medium.

\subsection*{Contributions of the Paper}

The main contributions of the present paper are 
\begin{enumerate}
\item \emph{\underline{Stochastic Space-Time Dynamic Model}}:
\newline

In Section \ref{sec:system_model}, we define precisely the mathematical model of the network along with the assumptions we are imposing for the mathematical analysis. This model is one of the contributions of the present paper as it captures precisely the stochastic interactions and dynamics \emph{both} in space and time. In Section \ref{sec:main_results} we state the main mathematical results of our paper. In subsection \ref{sec:results} we give  an \emph{exact} characterization of the time-ergodicity criterion i.e. give an explicit and simple formula to determine the phase-transition for dynamic stability. This notion of stability will be made precise in the sequel in Section \ref{sec:assumptions} . In section \ref{sec:clustering}, we  prove the intuitive result that, when it exists, the steady-state point process in our model exhibits a form of statistical clustering (made precise later), which is detrimental to performance as it creates higher interference powers at typical receivers than in a network with complete independence. We provide the proof of the ergodicity criterion in the Appendix ( Section \ref{sec:proofs}) which requires the use of point-process theory and in particular Palm calculus and stochastic coupling arguments. Our proof techniques for handling dynamic point-processes are  to the best of our knowledge new and potentially useful for analyzing other similar dynamic models of wireless networks. More generally, from an information theoretic perspective, we exhibit a form of dynamic network capacity when treating interference as noise. Our framework could potentially be extended to consider the dynamic network capacity for potentially other channels as well. From a queuing perspective, we exhibit through our model, a new form spatial queuing which cannot be reduced to any traditional non-spatial queuing network. These viewpoints of our model allows us to pose many more different problems than can be answered in this paper.
\\

\item \emph{\underline{ Formulas for Delay and System Design Insights}}:
\newline

We provide an explicit closed form formula to compute the phase-transition for dynamic stability in Section \ref{sec:main_results}. The phase-transition result however only provides whether the  delay experienced by a typical link  is finite or not. In Section \ref{sec:perf_bound}, we propose two formulas to approximately compute the mean-number of links per-unit space and the average delay of a typical link. The simplest heuristic is a first order Poisson approximation which relies on a single intensity parameter and hence cannot take clustering into account. We also propose another heuristic, which is a second order cavity type approximation of the second moment measure \cite{daley} of the steady-state point process. We find through simulations, that this heuristic works very well in all regimes. This heuristic is potentially useful to derive explicit approximate formulas for mean delay in other spatio-temporal models. From a practical networking perspective, closed form expressions for delay based on system parameters is very crucial. The formulas for delay provide insight into how to dimension D2D networks in terms of maximum allowable space-time traffic intensity or minimum spectral bandwidth needed to provide mean-delay based guarantees to the links in the ad-hoc network.

%
\end{enumerate}

\section{System Model - Birth-Death Model for Wireless Flows}
\label{sec:system_model}

In this section, we  describe the mathematical model of the dynamic wireless network which we later analyze. Roughly speaking, our model of a network is one wherein links which are transmitter-receiver pairs arrive into the network which is Euclidean space. Each transmitter of a link has a file it wants to send to its receiver. The speed or rate at which a transmitter can send its file to the receiver is a function of the positions of other transmitters transmitting files to their respective receivers. Upon completion of file transfer, a link departs from the network. We make the above dynamic description of the network more precise in the sequel. In subsection \ref{subsec:Spatial_Domain}, we describe the continuum network topology. In subsection \ref{subsec:links_arrival}, we describe the process of link and traffic arrivals into the network. Subsection \ref{subsec:data_rate} gives the precise description of how the instantaneous speed or instantaneous rate of file transfer of a link is affected the presence of other transmitting links. Finally, in subsection \ref{subsec:dynamics}, we put together the preceding parts by  compactly describing the arrival-departure dynamics of the wireless links we consider in this paper.


\subsection{Spatial Domain}
\label{subsec:Spatial_Domain}
The wireless links considered in this setup are transmitter-receiver pairs. The network at any point of time consists of a certain number of transmitters each transmitting to its own unique intended receiver. This is also commonly referred to as the ``dipole-model'' of a D2D ad-hoc wireless network.
\\

 The wireless links live in $\mathbf{S} \subset \mathbb{R}^2 = [-Q,Q] \times [-Q,Q]$, a square region of the Euclidean plane where $Q$ is a large but fixed \emph{finite} constant. To avoid edge effects, we identify the opposite edges of the square and wrap it around to form a torus. We denote by $|\mathbf{S}|$ as the area of the region $\mathbf{S}$ which is $4Q^2$. We present the mathematical analysis assuming $\mathbf{S}$ is a square torus as it makes exposition of proof ideas easier.
 
%

\subsection{Links and Traffic Arrival Process}
\label{subsec:links_arrival}

The links arrive into the network as a stationary marked space-time process on $\mathbf{S} \times \mathbb{R}$ with intensity $\lambda$. This marked point-process on $\mathbf{S} \times \mathbb{R}$  is denoted by $\mathcal{A}$. An atom $p \in \mathbb{Z}$ of $\mathcal{A}$ represents the receiver and is denoted by $(x_p,b_p)$. $x_p \in \mathbf{S}$ denotes the spatial location of receiver $p$ and $b_p \in \mathbb{R}$ denotes the time of arrival into the network of receiver $p$. Hence, one can represent the point process $\mathcal{A}$ as $\mathcal{A} = \sum_{p \in \mathbb{Z}} \delta_{(x_p,b_p)}$, where $\delta_{(x,b)}$ refers to the Dirac-mass at $(x,b) \in \mathbf{S} \times \mathbb{R}$. To each point $p$ of $\mathcal{A}$, we associate a vector mark of $(y_p,L_p)$, where $y_p \in \mathbf{S}$ and $L_p \in \mathbb{R}^{+}$, where $y_p$ refers to the location of the transmitter of receiver $p$ and $L_p$ denotes the file-size which the transmitter of $p$ wants to send to the receiver of $p$. We refer to the pair $(x_p;y_p)$ as link $p$ whose receiver is in location $x_p$ and transmitter in location $y_p$. The length of link $p$ is denoted by $T_p := ||x_p - y_p||$.
\\

The set of links present or \emph{alive} in the network at time $t$ is denoted by $\phi_t$ i.e. $\phi_t = \{(x_1;y_1),...,(x_{N_t};y_{N_t}) \}$, where $N_t$ is the number of links alive in the network at time $t$. The exact dynamics describing which links are present at a particular time $t$ will be specified in the sequel. More formally, $\phi_t = \sum_{i=1}^{N_t}\delta_{x_t}$ is a point-process on $\mathbf{S}$ of receivers marked with the location of their transmitters. We  use the terminology ``\emph{configuration of links}'' to refer to a marked point-process on $\mathbf{S}$ (atoms representing the receiver locations) with its marks (representing its corresponding transmitter locations)  in $\mathbf{S}$. We denote by $\phi^{Tx}_{t} = \{y_1, \cdots , y_{N_t}\}$, the point-process of transmitters present at time $t$ in the network and by $\phi^{Rx}_{t} = \{x_1, \cdots , x_{N_t}\}$, the point process of receivers at time $t$ in the network. 
\\


This arrival process can be seen as an incarnation of links initiating  communication in a very dense IoT or a D2D network for instance. When a link has a file to transmit (which comes rarely and randomly in time), a node ``switches on'' and initiates contact with its receiver. Since the network is dense and arrivals are rare, the spatial locations of links initiating connection can be seen as coming from a space-time point-process which we model as the link arrival process.

\subsection{Data Rate}
\label{subsec:data_rate}
The transmitter of each link $p$ has a file of size $L_p$ measured in  bits which needs to be communicated to its receiver. The transmitter sends this file to its receiver at a time varying rate given by the instantaneous Shannon rate. Denote by $l(\cdot): \mathbb{R}^{+} \rightarrow \mathbb{R}^{+}$, a distance dependent `path-loss' function which encodes how  signal power attenuates with distance. More precisely, $l(r)$ is the received power at  distance $r$ from a transmitter transmitting at unit-power. We can thus, define the rate of file transmission by a transmitter to its receiver as   
\begin{align}
R(x,\phi) = C \log_2 \left( 1 + \frac{l(||x-y||)}{\mathcal{N}_0 + \sum_{u \in \phi^{Tx} \setminus \{y\}} l( ||x-u||)}\right).
\label{eqn:rate_defn}
\end{align}

In the above expression, $C$ is a constant with units in bits per unit time, $\mathcal{N}_0$ denotes the thermal noise power at the receiver, $ \sum_{u \in \phi^{Tx} \setminus \{y\}} l( ||x-u||)$ denotes the interference seen at location $x$ due to configuration $\phi$ and $l(||x-y||)$ is the received signal power at $x$ from $y$. The interference at location $x$ is the sum of attenuated powers from the transmitters in $\phi^{Tx} \setminus \{y\}$ which is the sum of attenuated powers from all other transmitters other than the transmitter of the tagged receiver under consideration. For any $(x;y) \in \phi$, denote by $I(x,\phi)$ as the interference seen at $x$ in configuration $\phi$, which can be written as 
\begin{align}
I(x,\phi) = \sum_{u \in \phi^{Tx}\setminus \{ y\}} l(||x-u||).
\label{eqn:interference_defn}
\end{align}
Further, denote by $a$ the constant (which can possibly be infinite) $a = \int_{x \in \mathbf{S}}l(||x||)dx$.
\\

Some common examples of path-loss functions are
\begin{itemize}
\item $l(r) = r^{-\alpha}$ with $\alpha > 2$ called the ``power-law path-loss'' model.
\item $l(r) = (r+k)^{-\alpha}$ where $k$ is a constant  is commonly called the ``bounded path-loss'' model.
\end{itemize}
In our analysis however, we remain general and do not explicitly assume a particular form for the function $l(\cdot)$. Equation (\ref{eqn:rate_defn}) is the Shannon formula for the Gaussian SISO (Single Input Single Output) channel  with signal power $1$ and the interference treated as noise  \cite{tom_cover}. We will comment on extensions of the dynamics to MIMO channels in Section \ref{sec:mimo_extnsions}.
\\

In Equation (\ref{eqn:rate_defn}), we did not consider the effect of random channel fading. However, one can easily model the effect of fast fading by defining the rate-function as  
\begin{align}
R^{(f)}(x,\phi) = C\mathbb{E}_h \log_2 \left(  1 + \frac{h_{xy} l(||x-y||) }{\mathcal{N}_0 + \sum_{t \in \phi^{T} \setminus \{ y\} } h_{xt} l(|| t - x||) } \right)  ,
\label{eqn:random_rate}
\end{align}
where $h_{xy}$ and $h_{tx}$ are independent random-variables representing the values of the fading power between the different transmitters and receiver and the expectation is with respect to this random vector of fades $h$. All of our theoretical results extend to this case but with a bit more notation and computation cost and thus, we discuss only the case without fading. The reason for fast-fading to not affect our theoretical insights is that both Equations (\ref{eqn:rate_defn}) and (\ref{eqn:random_rate})  are deterministic monotone functions of the point $x$ and $\phi$. The rate functions are monotone in the sense that if $(x;y) \in \phi_1 \subseteq \phi_2$, then we have $R(x,\phi_1) \geq R(x,\phi_2)$ and $R^{(f)}(x,\phi_1) \geq R^{(f)}(x,\phi_2)$. We see from the proofs of our results, that these two (monotonicity and deterministic) are the crucial aspects of rate function on which the results hinge on and hence, we will only discuss the case without fading to simplify notation and convey the main ideas.

\subsection{The Dynamics}
\label{subsec:dynamics}
This setup now allows one to precisely define the network dynamics. A link arriving with receiver in location $x_p \in \mathbf{S}$ and its transmitter at location $y_p \in \mathbf{S}$ at time $t_p$ with file of size $L_p$ leaves the network at time $d_p$ given by the following recursive definition
\begin{equation}
d_p = \inf \left\{ t > b_p : \int_{u = t_p}^{t} R(x_p, \phi_u) du \geq L_p \right\}.
\label{eqn:dynamics_defn}
\end{equation}
In the above equation, $\phi_u$ denotes the point process of all links ``alive'' at time $u$ i.e. $\phi_{u}^{R} = \sum_{p \in \mathbb{Z}} \delta_{x_p} \mathbf{1}_{\{ u \in [b_p,d_p]\}}$ and $\phi_{u}^{T} = \sum_{p \in \mathbb{Z}} \delta_{y_p} \mathbf{1}_{\{ u \in [b_p,d_p]\}}$ where $\delta_x$ denotes to the Dirac-measure at location $x \in \mathbf{S}$. We refer to the time instant $b_p$ as the ``birth'' time of link $p$ and $d_p$ as the ``death'' time of link $p$. This is the justification for calling this dynamics a ``spatial birth-death'' model, i.e. this transmitter-receiver pair is ``born'' at time $b_p$ and ``dies'' at time $d_p$ and leaves the network.
\\

This model is the wireless analog of the ``flow-level" model introduced by Massoulie and Roberts \cite{BW_Sharing} to evaluate and study wired networks, particularly the Internet. The flow-model in the present paper is based on a more precise modeling of the wireless interactions compared to the standard conflict graph model of interference.  This spatial birth-death model can also be viewed as a ``dynamic" version of the model considered in \cite{baccelli_gammal}, namely the Gaussian Interference channel with point-to-point  codes. In our model,  each link or a ``flow" is a Guassian point-to-point channel using a point-to-point codebook and treats all Interference as Noise (IAN) as made evident in the rate-formulation in Equation \ref{eqn:rate_defn}. It was shown in \cite{baccelli_gammal}, that one can consider other schemes such as Successive Interference Cancellation or Joint Optimal Decoding to get strictly better performance than considering Interference as Noise in cases of static links that use ptp codes. We however only study the dynamic version of treating IAN  and leave the other cases for future work.

%
%
%

%

\subsection{Mathematical Assumptions}
\label{sec:assumptions}

All the analysis and results rely on the following assumptions on the system model presented in the previous section. 
\begin{enumerate}
\item The link arrival process is a time-space stationary Poisson Point Process of intensity $\lambda$. The probability of an arrival of a receiver in an infinitesimal location $dx$ in an infinitesimal time interval $dt$ is $\lambda dx dt$.
\item The file sizes of each transmitter are i.i.d. and exponentially distributed with mean $L$ bits.  
\item The transmitter location $y$ of a receiver at $x$ is assumed to be distributed uniformly and independently of everything else on the perimeter of a ball of radius $T$ centered at $x$. In particular, the received signal power at any receiver is $l(T)$. 
\item The thermal noise power $\mathcal{N}_0 >0$ is a fixed constant.
\item The path-loss function is bounded and non-increasing with $l(0) = 1$. This is a reasonable assumption since energy is only dissipated on traveling through space and the received energy can be no larger than the transmit energy.
\end{enumerate}

These assumptions (especially the statistical ones) are imposed primarily for mathematical tractability.  It is well known, at least in the context of the Internet, that file sizes are Pareto \cite{Internet_File} and it would make modeling sense to assume heavy-tailed file sizes.  We will relax the statistical assumption on exponential file-sizes in the simulation studies. Nonetheless, studying the system under the Markovian statistical assumptions form a necessary first step before considering the general case. 
\\

In our model, we have that all links have the same length of $T$. This is commonly referred to as the `Dipole-Model' of an ad-hoc wireless network \cite{Baccelli_sg_book}. An interesting limiting case is that of $T=0$. This corresponds to the physical case of when the link lengths are very small compared to the size of the network.  In this limiting case, the point process $\phi_t$ is simple and unmarked since the transmitter and receiver locations are identical, and the signal power is $l(0) = 1$. The interference function at a point $x$ from configuration $\phi$ is then $I(x,\phi) = \sum_{y \in \phi \setminus\{x\}}l(||y-x||)$. We mention this limiting case here as it will help us to get a much better understanding of what our theoretical results imply, especially that of clustering (defined later in Definition \ref{defn:clustering}). However, all of our mathematical results are valid for general arbitrary link distances $T$. 
\\

Although the  assumptions may render the model somewhat specific, it still presents a formidable mathematical challenge and captures the key features of a spatio-temporal dynamic wireless network. Most prior works incorporating spatial interference circumvent this mathematical difficulty by making `full-buffer' assumptions which is equivalent to assuming no temporal interactions.  Our results, especially the closed form expressions for approximating of delay are the first in the context of spatio-temporal wireless network models to the best of our knowledge.
\\

The  statistical assumptions, namely the Poisson arrival process and i.i.d. exponential file sizes imply that the process $\phi_t$ is a continuous time measure-valued Markov Chain on the state space of marked simple counting measure on $\mathbf{S}$ denoted as $\mathbf{M}(\mathbf{S})$ \cite{daley}. More precisely, the process $\phi_t$ is a piece-wise constant jump Markov Process i.e., from a time $t$, the next \emph{change} in the configuration will occur after an exponentially distributed time duration with rate $\lambda |S| + \frac{1}{L}\sum_{x \in \phi_t}R(x,\phi_t)$. This interpretation follows since births occur at the epochs of an exponential clock with rate $\lambda |S|$ and the death rate of any receiver $x$ in configuration $\phi$ is $\frac{1}{L}R(x,\phi)$ which is independent of everything else. The assumption $Q < \infty$ ensures that $\phi_t$ is a piece-wise constant jump process. Extending the analysis of stability to the case of $\mathbf{S} = \mathbb{R}^2$ is way more challenging and is left for future work. The large torus is meant to emulate the Euclidean space. The fact that it is similar to the Euclidean space (in terms of interference field and hence birth and death dynamics) justifies our use of the Palm calculus of the Euclidean space rather than that of the torus in some derivations.
\\

The first natural question we ask about $\phi_t$ is that of time ergodicity which we address in the next section. Time ergodicity implies that the process $\phi_t$ admits an unique steady-state in which the links form a stationary and space-ergodic point process on $\mathbf{S}$. Moreover, since $\mathbf{S}$ is a compact set, the stationary-regime when it exists will put only finitely many points in $\mathbf{S}$ at any given instant almost-surely. Denote by $\phi_0$  the steady-state point-process of links i.e. the links that are ``alive'' or active in steady-state. $\phi_0$ is a point-process on $\mathbf{S}$ with atoms representing the locations of receivers and marks representing the relative transmitter locations.  
\\

 Denote by $\beta$ the density of links present in the network in steady-state (assuming it exists). More formally, $\beta$ denotes the intensity of the receiver point-process $\phi_{0}^{Rx}$ (which is the ground point process of $\phi_0$)  on $\mathbf{S}$ when the dynamics is in steady state. Note that the intensity of the transmitter point-process $\phi_{0}^{Tx}$ in steady-state is also $\beta$ since every receiver in the model has exactly one transmitter. The distribution of the relative location of the transmitter of a typical receiver of $\phi_{0}^{Rx}$ is uniform on the perimeter of a ball of radius $T$ around this receiver. However, the transmitter locations across different receivers of $\phi_{0}^{Rx}$ are not independent due to the correlation (clustering) induced by the dynamics.
\\

 The interpretation of time ergodicity is also connected to the phase-transition of mean delay. Little's law for this dynamics yields $\beta = \lambda W$, where $W$ is the average sojourn time of a typical link i.e., $W = \mathbb{E}[d_0 - b_0]$; which follows from PASTA \cite{pasta}. The process $\phi_t$ being time ergodic in our model is equivalent to asserting that $W < \infty$, i.e. finite mean delay for a typical link in the network. This interpretation is what we allude to in the system insight section which allows one to evaluate how frequently in space and time should the traffic arrival process be (i.e. how large $\lambda$) can be for the network to provide finite mean-delay to all links.

\begin{table*}
\let\centering\relax
\resizebox{\linewidth}{!}{%
\begin{tabular}{||c|l||}
\hline
Notation & Description \\
\hline \hline
$\phi_t$ & Point-process of receivers alive at time $t$ marked with their corresponding transmitter locations \\
\hline
$\phi_t(\mathbf{S})$ & The number of links alive at time $t$\\
\hline
$\phi_{t}^{Tx}$ & The point-process of transmitter locations at time $t$ \\
\hline
$\phi_{t}^{Rx}$ & The point-process of receiver locations at time $t$ \\
\hline
$\phi_0$ & The steady state marked point-process corresponding to $\phi_t$ \\
\hline
$\phi_0(\mathbf{S})$ & The random variable denoting the number of links in steady-state \\
\hline
$\mathbb{E}^{0}_{\phi_0}$ & The Palm probability measure with respect to $\phi_0$ \\
\hline
$K_{\phi}(r)$ & The Ripley K-function for point process $\phi$ \\
\hline
$T$ & The link length in the model \\
\hline
$\beta$ & The intensity of the point process $\phi_0$ \\
\hline
$a$ & $\int_{x \in \mathbf{S}}l(||x||)dx$ \\
\hline
$C$ & Multiplicative constant for the rate-function in Equation (\ref{eqn:rate_defn}). It is measured in bits per second.\\
\hline
$L$ & Average file size measured in bits.\\
\hline
\end{tabular}}
\label{table_notation}
\caption{Table of Notation}
\end{table*}

\section{Main Theoretical Results}
\label{sec:main_results}

The main theoretical results of our paper are on the time-ergodicity (or stability) conditions of the dynamics $\phi_t$ and on a certain structural characterization of the steady-state point process of $\phi_t$ whenever it exists.  The proofs of the theorems are presented in the Appendix.

\subsection{Stability Criterion}
\label{sec:results}


%

We state our main theoretical results on the stability criterion (i.e. time ergodicity) of the dynamics.

\begin{theorem}
If $\lambda > \frac{C l(T)}{\ln (2) L a}$, then the Markov Chain $\phi_t$ admits no stationary regime.
\label{thm:necessary_condn}
\end{theorem}

We see from the proof (in Section \ref{sec:proof1}) that this theorem only needs the weaker assumption  that $l(\cdot)$ be such that $l(r) < \infty$ for all $r>0$. This indeed is a weaker assumption than assuming that the function $l(\cdot)$ is bounded. Thus, we have as immediate corollary to this theorem:
\begin{corollary}
For the path-loss model $l(r) = r^{-\alpha}$, $\alpha \geq 2$, 
for all $\lambda > 0$, and all mean file sizes, 
the process $\phi_t$ admits no stationary-regime. 
\end{corollary}
\begin{proof}
This follows since the integral $\int_{x\in\mathbf{S}}l(||x||)dx$ diverges for the function $l(r) = r^{-\alpha}$ for all $\alpha \geq 2$.
\end{proof}

 The next result provides a tight condition for time ergodicity.
\begin{theorem}
If $\lambda < \frac{C l(T)}{\ln(2) L a}$, then the Markov Chain $\phi_t$ is time ergodic, i.e. has an unique stationary regime. 
\label{thm:sufficient_condn}
\end{theorem}

We note that the above theorems statements are valid as is even in the case of fading if one used the rate-function in Equation (\ref{eqn:random_rate}) with the fades being unit-mean i.i.d. random variables. The two theorems identify the exact critical arrival rate $\lambda$ for ergodicity as $\lambda_c = \frac{C l(T)}{L \ln(2) a}$. We however refrain from studying the critical case as it is technically more subtle. In the sequel, whenever we refer to $\phi_0$, we implicitly assume $\phi_t$ is ergodic, i.e. the condition $\lambda < \frac{C}{\ln(2) L a}$ holds.
%

\subsection{Clustering}
\label{sec:clustering}

In this section, we state the main structural characterization of the steady-state point process $\phi_0$ when it exists i.e. when $\lambda <  \frac{C}{L \ln(2) a}$. We need the following definition of clustering.

 \begin{definition} (CLUSTERING)
Let $\phi$ be a stationary configuration of links, i.e. it is a stationary marked point-process on $\mathbf{S}$ with its marks in $\mathbf{S}$. Then $\phi$ is said to be clustered if for all bounded, positive, non-increasing functions $f(\cdot) : \mathbb{R}^{+} \rightarrow \mathbb{R}^{+}$, the following inequality holds
\begin{equation}
\mathbb{E}^{0}_{\phi}[F(0,\phi)] \geq \mathbb{E}[F(0,\phi)]  ,
\label{eqn:thm_clustering}
\end{equation}
where $F$ is the shot-noise defined as follows. For any atom (receiver) $x \in \phi$ with its corresponding mark (transmitter) $y \in \mathbf{S}$, the shot noise $F(x,\phi) := \sum_{T \in \phi^{Tx} \setminus\{y\}}f(||T-x||)$.
\label{defn:clustering}
 \end{definition}
 
 
 \begin{theorem}
If the dynamics $\phi_t$ is ergodic, then the steady state point process $\phi_0$ is clustered.
\label{thm:clustering}
 \end{theorem}
%

By substituting $f(\cdot) = l(\cdot)$ in Equation (\ref{eqn:thm_clustering}) , we get that the mean of the interference measured at any uniformly randomly chosen receiver in the steady-state point process (this is the interpretation of the Palm probability) is larger than the mean of the interference measured at any uniformly randomly chosen location of space in $\mathbf{S}$.
\\

To  understand  why the above definition is a form of clustering, consider the case $T=0$ which gives a clearer picture.  In this case, Theorem (\ref{thm:clustering}) gives a clustering comparison of $\phi_0$ with a Poisson Point Process (PPP) of same intensity. Let $\psi$ be a PPP of the same  intensity as $\phi_0$. Then, from Slivnyak's theorem (Theorem $1.4.5$, \cite{Baccelli_sg_book}), one can rewrite the inequality in (\ref{eqn:thm_clustering}) as 
\begin{align}
\mathbb{E}^{0}_{\phi_0}[F(0,\phi_0)] \geq \mathbb{E}^{0}_{\psi}[F(0,\phi_0)],
\end{align}
where $\mathbb{E}^{0}_{\psi}[F(0,\phi_0)] = \mathbb{E}[F(0,\psi)]$ follows from Slivnyak's theorem which  is equal to $\beta \int_{x \in \mathbf{S}}f(||x||)dx$ from Campbell's Theorem (Theorem $1.4.3$, \cite{Baccelli_sg_book}). Slivnyak's theorem essentially gives that the PPP has no clustering i.e. the Inequality  \ref{eqn:thm_clustering}  is an equality. Hence, we automatically have a shot noise comparison of the steady state point process $\phi_0$ with a PPP. 
\\

The comparison with a PPP also gives us a comparison of the Ripley K-function \cite{spatial_stats} of $\phi_0$ with that of a PPP. The Ripley K-function $K_{\phi}(\cdot) : \mathbb{R}^{+} \rightarrow \mathbb{R}^{+}$ of a point-process $\phi$ is defined as $K_{\phi}(r) = \frac{1}{\beta}\mathbb{E}^{0}_{\phi}[ \phi(B(0,r))-1] $ where $\beta$ is the intensity of $\phi$ and $\mathbb{E}^{0}_{\phi}$ is the Palm probability measure of $\phi$. This function  can be interpretted as  the mean number of points (scaled by the intensity of the point-process) within distance $r$ to the origin \emph{conditioned} on a point of $\phi$ to be present at the origin. The Ripley K-function is commonly used in statistical analysis of point-patterns to identify if an empirical data-set exhibits statistical clustering \cite{spatial_stats}. Based on the shot-noise comparison with a PPP, we have the following corollary.
\begin{corollary}
Assume $\phi_t$ is in steady-state and $T=0$. Denote by $\beta$ to be the intensity of $\phi_0$ and $\psi$ to be a PPP on $\mathbf{S}$ with intensity $\beta$. Then, $K_{\phi_0}(r) \geq K_{\psi}(r)$.
\label{cor:ripley}
\end{corollary}
\begin{proof}
Consider $f(x) = \mathbf{1}(x \leq r)$ in Theorem \ref{thm:clustering}. 
\end{proof}
We will use Ripley K-function in the simulations to compare the point process $\phi_0$ with a PPP to derive a bound on the intensity $\beta$ of $\phi_0$ as a function of $\lambda$, $L$ and  $l(\cdot)$. 
\\

%
%

Intuitively, it is not surprising to expect a clustered point-process in steady state. An arriving link gets lower rate if it is in a crowded area of transmitters, due to interference. This arriving link also causes more interference to the cluster of links already present thereby causing more interference and slowing everyone down. This reinforcement of service slowdown is actually the fundamental reason making the system always unstable in the power law attenuation function case. More generally, when $\phi_t$ is sampled in steady-state, it is expected to be clustered as formalized by Theorem \ref{thm:clustering}. A snapshot of the point-process $\phi_0$ is presented in Figure \ref{fig:clustering_points} which gives a visual illustration of the clustering.

\section{Performance Analysis - Steady State Formulas}
\label{sec:perf_bound}

In this section, we propose two heuristic formulas for $\beta$ the intensity of the point process $\phi_0$ as a function of $\lambda$. Note that a heuristic formula for $\beta$   gives a heuristic formula for mean delay $W$ through Little's Law  ($\beta = \lambda W$).
\\

 We propose two formulas -  $\beta_f$ called the \emph {Poisson Heuristic} and  $\beta_s$ called \emph{Second-Order heuristic} to approximate $\beta$ the intensity of the steady-state point process $\phi_0$. We show that subject to a natural conjecture (Conjecture \ref{conjecture_1}),  $\beta_f$ is a lower bound on $\beta$. We see from simulations however that $\beta_s$ is a much better approximation of  $\beta$ compared to $\beta_f$.  Both  formulas are derived based on  approximately evaluating  the following Equation which we establish in Equation (\ref{eqn:proof1_RCL_bits2}) in the Appendix.
\begin{align}
\lambda L = \beta \mathbb{E}^{0}_{\phi_0}\left[ \log_2 \left( 1 + \frac{l(T)}{\mathcal{N}_0 + I(0,\phi_0)}\right)\right].
\label{eqn:simulation_exact}
\end{align}


\subsection*{The Poisson Heuristic}

The Poisson heuristic formula $\beta_f$ is given by the largest solution to the following fixed point equation
\begin{equation}
\lambda L = \frac{\beta_f}{\ln(2)}\int_{z=0}^{\infty}\frac{e^{-\mathcal{N}_0 z}(1 - e^{-z l(T)})}{z} e^{- \beta_f \int_{x \in \mathbf{S}}  (1 - e^{-zl(||x||)} ) dx} dz.
\label{eqn:simulation_heuristic}
\end{equation}
This formula is obtained by approximating the expectation in  Equation (\ref{eqn:simulation_exact}) by assuming the following ``\emph{Independent Poisson heuristic}''.  We assume that $\phi_0$  is an independently marked Poisson-Point process with the transmitter locations of different receivers in $\phi_0$ being independent. Since the transmitter locations are assumed to be independent, the process $\phi_{0}^{Tx}$ will also be a PPP in this Poisson heuristic.  We state the following lemma without proof from \cite{Interference_Lemma} which is useful in computing the expectation under the Poisson assumption.

\begin{lemma}
Let $X,Y$ be non-negative and independent Random Variables. Then,
\begin{equation}
\mathbb{E}\left[\ln \left(1 + \frac{X}{Y+a}\right) \right] = \int_{z=0}^{\infty}\frac{e^{-az}}{z}(1 - \mathbb{E}[e^{-zX}])\mathbb{E}[e^{-zY}]dz. \nonumber
\end{equation}
\label{lem:simulation_integral}
\end{lemma}

We can then explicitly compute the expectation in Equation (\ref{eqn:simulation_exact}) by letting $X = l(T)$ to be deterministic and $Y = I(0,\phi_0)$ as follows

\begin{align}
\lambda L &= \beta_f \mathbb{E}^{0}_{\psi}\left[ \log_2 \left( 1 + \frac{l(T)}{\mathcal{N}_0 + I(0)}\right)\right]  \nonumber \\
&\stackrel{(a)}{=} \beta_f\mathbb{E}_{\psi}\left[ \log_2 \left( 1 + \frac{l(T)}{\mathcal{N}_0 + I(0)}\right)\right] \nonumber \\
& \stackrel{(b)}{=} \frac{\beta_f}{\ln(2)}\int_{z=0}^{\infty}\frac{e^{-\mathcal{N}_0 z}(1 - e^{-z l(T)})}{z} e^{- \beta_f \int_{x \in \mathbf{S}}  (1 - e^{-zl(||x||)} ) dx} dz, \nonumber
\end{align}
where $\psi$ is a Poisson Point Process on $\mathbf{S}$ with intensity $\beta_f$. The equality $(a)$ follows from Slivnyak's theorem and the equality $(b)$ follows from Lemma \ref{lem:simulation_integral} and the formula for the Laplace functional of a Poisson Point Process. The subscript $f$ refers to the computation of the density under this Poisson heuristic. This establishes the formula in Equation (\ref{eqn:simulation_heuristic}).

We now make the following conjecture on the higher-order moment measures of $\phi_0$, which we will  leverage to show that $\beta_f$ is a lower bound on $\beta$.

\begin{conjecture}
Let $\phi_0$ be the point process on $\mathbf{S}$ corresponding to the stationary distribution of $\phi_t$ with intensity $\beta$. Denote by $\psi$ to be an independently marked Poisson Point Process on $\mathbf{S}$ with intensity $\beta$. The mark of any atom $x$ of $\psi$ is a point $y$ drawn uniformly on the perimeter of a circle of radius $T$ around $x$. Then, for any $s>0$, we have $\mathbb{E}^{0}_{\phi_0}[e^{-s I(0;\phi_0)}] \leq \mathbb{E}^{0}_{\psi}[e^{-s I(0;\psi)}] $.
\label{conjecture_1}
\end{conjecture}
Note that from Slivnyak's theorem we also have $\mathbb{E}^{0}_{\psi}[e^{-s I(0;\psi)}]  = \mathbb{E}_{\psi}[e^{-s I(0;\psi)}]$. This conjecture which is validated through simulations in Figure \ref{fig:conjecture_1}, is a slightly different statement on the structural characterization of $\phi_0$ than stated in Theorem \ref{thm:clustering}. This conjecture gives that the Laplace transform of the interference measured at a typical receiver of $\phi_0$ is larger than that measured at a typical receiver of an equivalent PPP. In general, whenever we have ordering of the mean, then we have ordering of the Laplace Transform only as $s \rightarrow 0$. This ordering for the Laplace transform as $s \rightarrow 0$ follows from Taylor's expansion that $e^{-sx} \approx 1 - sx$ as $s \rightarrow 0$.  However, in our case, we believe that the ordering on the Laplace transform holds for all $s \geq 0$ but we are unable to prove so. The intuition for this follows from the pictorial interpretation that there are roughly the same number of interfering transmitters around a typical receiver in $\phi_0$ and $\psi$ since the intensities of $\phi_0$ and $\psi$ are the same. However, the interfering transmitters are closer to the typical receiver in $\phi_0$ as compared to in $\psi$. This intuition follows from Corollary \ref{cor:ripley} where we had ordering of the Ripley-K function of $\phi_0$ and $\psi$. This pictorial interpretation then gives an intuition for the conjecture since, the interference is  the sum of attenuated powers from interfering transmitters where the attenuation is through a function that is non-increasing with distance. Thus, $I(0)$ is the sum of roughly the same number of terms in both $\phi_0$ and in $\psi$, but each of the terms are slightly larger in $\phi_0$ than in $\psi$. This interpretation can possibly be made rigorous in the asymptotic regime as  $\lambda \uparrow \lambda_c$ by alluding to certain concentration phenomenon. However, we see from simulations that this conjecture holds true for all regimes of $\lambda$. This conjecture is further substantiated in Figure (\ref{fig:perf_comparision}) which underpins Proposition \ref{proposition_1}. 
\\

The ordering of the mean does not always imply the ordering of Laplace transforms in general. As a very simple example consider two random variables $X$ and $Y$ where $X$ takes values $\{1,2,3,4\}$ with probabilities $\left\{ \frac{1}{6}, \frac{1}{3}, \frac{1}{6}, \frac{1}{3}   \right\}$ and $Y$ is deterministic and takes value of $2$. Here $\mathbb{E}[X] = \frac{8}{3}$ and $\mathbb{E}[Y] = 2$. However, for $s=1.1$, $\mathbb{E}[e^{-sX}] > \mathbb{E}[e^{-sY}]$. More generally if $\mathbb{E}[X] \geq \mathbb{E}[Y]$ but the higher order moments are ordered in the opposite direction, then one cannot expect an ordering on the Laplace-transform.

\begin{figure}
\centering
\includegraphics[scale=0.35]{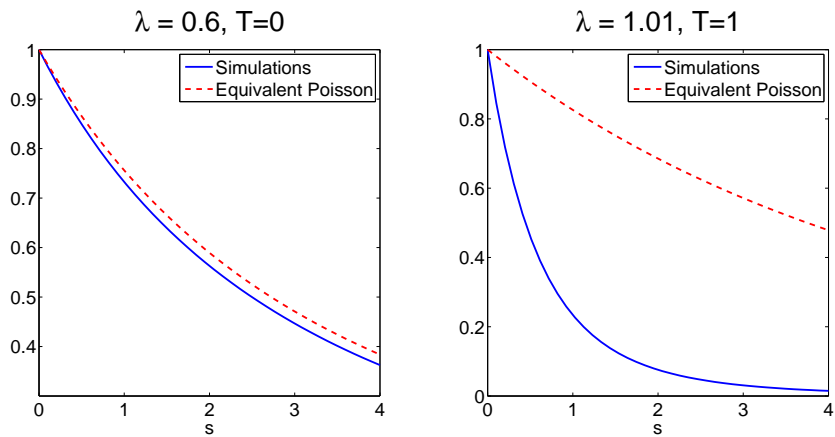}
\caption{A plot comparing the functions $\mathbb{E}^{0}_{\phi_0}[e^{-sI(0;\phi_0)}]$ and $\mathbb{E}^{0}_{\psi}[e^{-sI(0;\psi)}]$, for $l(r) = (r+1)^{-4}$.}
\label{fig:conjecture_1}
\end{figure}
%

\begin{proposition}
Subject to Conjecture (\ref{conjecture_1}), we have that $\beta \geq \beta_f$, where $\beta_f$ is the largest solution of  Equation (\ref{eqn:simulation_heuristic}).
\label{proposition_1}
\end{proposition}

\begin{proof}
Let $g(\beta) = \beta \mathbb{E}^{0}_{\phi_0}[R(0;\phi_0)]$ (where $\phi_0$ has intensity $\beta$) and let $p(\beta) = \beta \mathbb{E}^{0}_{\psi}[R(0;\psi)]$ where $\psi$ is a PPP on $\mathbf{S}$ with intensity $\beta$. Rate-conservation equation (\ref{eqn:simulation_exact}) gives that $\lambda L = g(\beta)$ and our heuristic computation is $\lambda L = p(\beta_f)$. From our conjecture and Lemma \ref{lem:simulation_integral}, we have the inequality $g(\beta) \leq p(\beta)$. The function $g(\beta) = \beta \mathbb{E}^{0}_{\phi_0}[R(0;\phi_0)]$ is monotone non-decreasing in $\beta$ as it describes the true dynamics through the equation $\lambda L = g(\beta)$. The monotonicity of $g(\cdot)$ along with the inequality $g(\beta) \leq p(\beta)$ gives the performance bound $\beta \geq \beta_f$.
\end{proof} 

 Proposition \ref{proposition_1}  gives that $\beta_f|\mathbf{S}|$ is a lower bound on the mean number of links present in the network in steady state and $\frac{\beta_f}{\lambda}$, as a lower bound on mean-delay of a typical link. 
\\

The Poisson heuristic completely ignores the spatial clustering we established in Theorem \ref{thm:clustering} and assumes complete-spatial randomness. Since it does not account for the clustering  it underestimates the typical interference seen at a receiver and therefore predicts a lower density of links. We see through simulations, that this heuristic is poor (i.e. the gap between $\beta$ and $\beta_f$ is large) in certain traffic regimes (Figure \ref{fig:perf_comparision}). This is not surprising as one cannot neglect the effect of spatial correlations  except in  asymptotic regimes of heavy and light-traffic (detailed later). Motivated by the poor performance of the Poisson heuristic in certain regimes,  we propose a ``second-order heuristic'' $\beta_s$ which takes into account the spatial correlations by considering an approximation of the second-order moment measure of $\phi_0$. We see through simulations (Figure \ref{fig:perf_comparision}) that $\beta_s$ is a much better approximation of $\beta$ than $\beta_f$ in all traffic regimes. 


\begin{figure} 
\centering 
\includegraphics[scale=0.5]{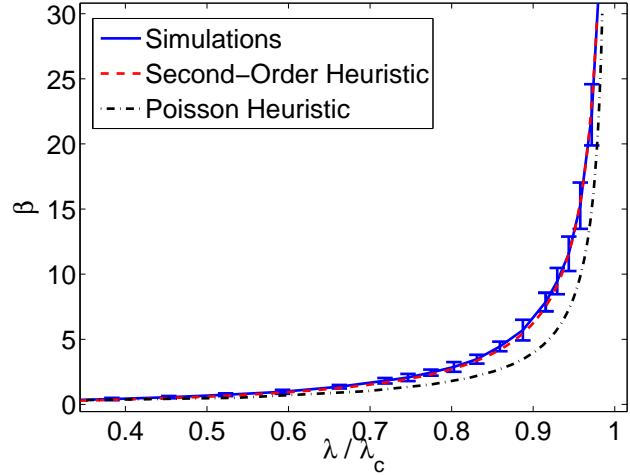}
\caption{The performance plot with $95\%$ confidence interval when $T=0$ and $l(r) = (r+1)^{-4}$. }  
\label{fig:perf_comparision}
\end{figure}

\subsection*{Second-Order Heuristic}

We propose a  heuristic formula $\beta_s$  for approximating $\beta$ in Equation (\ref{eqn:perf_heuristic}). For all values of $T$,  $\beta_s$  is given by
\begin{equation}
\beta_s = \frac{\lambda L}{C \log_2 \left(1 + \frac{l(T)}{\mathcal{N}_0 + {I_{s}} } \right)},
\label{eqn:perf_heuristic}
\end{equation}
where ${I_{s}}$ is the smallest solution of the fixed-point equation
\begin{equation}
{I_{s}} = \lambda L\int_{x \in \mathbf{S}} \frac{ l(||x||)}{C\log_2 \left( 1 + \frac{l(T)}{\mathcal{N}_0 + {I_{s}} + l(||x||)} \right)} dx.
\label{eqn:fixed_point}
\end{equation}

We call the heuristic in Equation (\ref{eqn:perf_heuristic}) a \emph{second-order heuristic} since it follows from an approximation of the second-order moment measure of $\phi_0$ as follows. Let ${I_{s}}$ denote the  interference of a typical point at $\phi_0$ and \emph{assume} it is non-random and equal to its mean. Then, Equation (\ref{eqn:perf_heuristic}) follows from Rate-Conservation  in Equation (\ref{eqn:simulation_exact}). To compute $I_s$, we  use the following approximation of the second order moment measure $\rho^{(2)}(x,y)$ of $\phi_0$ as 
\begin{equation}
\rho^{(2)} (x,y) \approx \frac{\beta \lambda L}{   C \log_2 \left( 1 + \frac{l(T)}{\mathcal{N}_0 + {I_s} + l(||x-y||)} \right)}.
\label{eqn:moment_measure_approx}
\end{equation}

Intuitively, the approximation is a form of cavity approximation which can be understood as follows. Two points at locations $x$ and $y$ will each ``see'' an interference of $I_s$ which is the interference of a typical point plus the additional interference caused by the presence of the other point. Using the above interpretation of interference, Equation (\ref{eqn:moment_measure_approx}) is a form of Rate-Conservation on the pair of points at $x$ and $y$. The average increase of the pair happens at rate $2 \lambda \beta$ and the average decrease of the pair happens at the rate equal to the sum of rates (since file-sizes are i.i.d. exponential) received by points $x$ and $y$ which is approximately $2(C/L) \log_2 \left( 1 + \frac{l(T)}{\mathcal{N}_0 + I_s + l(||x-y||)} \right)$ from the cavity approximation. Now, using the fact that $\mathbb{E}^{0}_{\phi_0}[I_0] := I_s = \frac{1}{\beta} \int_{x \in \mathbf{S}} l(||x||)\rho^{(2)}(x,0) dx$, we get Equation (\ref{eqn:fixed_point}) from Equation (\ref{eqn:moment_measure_approx}).
\\

The heuristic $\beta_s$  is compared against the true $\beta$ and the Poisson heuristic $\beta_f$ in Figure \ref{fig:perf_comparision}. The second-order heuristic performs much better compared to the Poisson-heuristic as it takes into account some notion of spatial correlations which the Poisson heuristic completely ignores.

\section{Simulation Studies}

\label{sec:simulations}


We perform simulations to gain a finer understanding of our model. We see that the bound in Proposition \ref{proposition_1} is tight in the two asymptotic regimes of light and heavy-traffic where the effect of spatial correlations vanishes.
We also argue that, in these two asymptotic regimes, the heuristic $\beta_s$ is ``close'' to $\beta_f$ thereby implying that  $\beta_s$ is also a good approximation to $\beta$. As noticed in Figure \ref{fig:perf_comparision}, $\beta_f$ is a poor approximation to $\beta$ compared to $\beta_s$ in the intermediate traffic-regime which we further highlight in this section. 
\\

To  explore the impact of spatial correlations, we study the tails of delay of a typical link and correlation between delays of different links in space. We observe that our model exhibits marked difference in terms of tail delay behavior from that of an equivalent queuing system which is obtained by a ``spatial-fluid'' approximation. We conclude from these studies that although our model resembles that of a queue (for e.g. the dynamics satisfies Little's Law), there are significant differences due to the spatial correlations, which in hindsight is not so surprising. We finally perform simulations with heavy-tailed file size distribution  and observe qualitatively the same phenomena as seen under exponential file-size distribution. We state our simulation results as claims which are not formal conjectures, but are  meant to provide a starting point for future research.

\subsection{Simulation Setup}

The path-loss function we consider  is $l(r) = (r+1)^{-4}$. Although all of the results qualitatively hold for any bounded-non-increasing function, we choose this power law function due to its wide-spread popularity in modeling wireless propagation. We assume unit link-length $T=1$  unless otherwise mentioned. We however note that all the qualitative results  carry over for any value of $T$ including the case of $T=0$. The pictures of point-process and the Ripley K-functions we test are those corresponding to the receivers.

\subsection{Tightness of $\mathbf{\beta_f}$}

We  study the bound in Proposition \ref{proposition_1} by empirically noticing how much $K_{\phi_0}$, the Ripley K-function of $\phi_0$, deviates from that of an equivalent PPP denoted by $K_{\mathrm{PPP}}$. The two Ripley K-functions being almost identical implies that the steady-state is ``almost'' Poisson and thereby the bound in Proposition \ref{proposition_1} is good. On the other hand, if there is significant deviation between the two Ripley K-function, then the bound is poor. We know from Corollary \ref{cor:ripley} that $K_{\phi_0}(r) \geq K_{\mathrm{PPP}}(r)$ for all $r \geq 0$. Here, we are interested in seeing how large  this difference can be.
\\

To plot the Ripley K-functions, we simulated the Markov chain $\phi_t$ in forward time for a long time to obtain a single sample of the steady-state $\phi_0$. We used the Spatstat package in R \cite{Spatstat} to perform spatial statistics and plot $K_{\phi_0}$. A single sample is sufficient as we take a large enough space (i.e. large $\mathbf{S}$) so that a single sample of $\phi_0$ has about $500$ points. Due to spatial ergodicity of $\phi_0$, we get a smooth estimate of the K-function from a single sample.
\\

 We observe  in Figures \ref{fig:clustering_1}, \ref{fig:clustering_2} and \ref{fig:clustering_3}, that the functions $K_{\phi_0}$ and $K_{\mathrm{PPP}}$ are very close in  heavy and light traffic and are very different in intermediate traffic. The heavy traffic  corresponds to the scenario when $\lambda$ is very close to the critical $\lambda_c$ and the light traffic  corresponds to the case when $\lambda$ is very ``close'' to $0$. We do not rigorously demarcate the exact space-time scaling needed to define the two asymptotic limiting regimes as it is beyond the scope of this paper. 
 
\begin{claim}
$\phi_0$ is almost Poisson in light-traffic. Moreover, the delay of a typical link converges weakly to an exponential distribution with mean $L \log_2 \left( 1 + \frac{1}{N_0}\right)^{-1}$ as $\lambda \rightarrow 0$. 
\label{claim_low}
\end{claim}

In the light-traffic regime, $\lambda$ is very ``small'' compared to $L$, and thereby  $\beta$ is also ``small''. This then implies that the distribution for the interference $I(0,\phi_0)$ is close to $0$, thereby making the interaction between the points almost negligible. The rate-function can then be approximated as $R(x,\phi) \approx \log_2 \left( 1 + \frac{1}{N_0}\right)$ and the dynamics   resembles that of a spatial $M/M/\infty$ queue whose stationary spatial distribution is  a PPP.  The intensity $\beta$ in this  regime is $\beta_{l} = \lambda \left(\log_2 \left( 1 + \frac{1}{N_0}\right)\right)^{-1}$. The subscript $l$ refers to the density computation in the interaction-less approximation. Figure \ref{fig:clustering_3} provides numerical evidence that  $\phi_0$ exhibits very little clustering in this regime and is ``close'' to a PPP. 
%
%

\begin{claim}
In the heavy-traffic regime, $\phi_0$ is almost Poisson, i.e. the effect of clustering vanishes as $\lambda \rightarrow \lambda_c$. 
\label{claim_heavy}
\end{claim}

The intuition behind the heavy-traffic behavior is that as $\lambda$ approaches $\lambda_c$, the stationary distribution is very dense, i.e. $\beta$ is large. Hence, the interference of a typical arriving link is mainly dominated by the local geometry which does not change much during the life-time of  the typical link. This indicates that the dynamics behaves very similarly to a heavily loaded $M/M/1$ Processor Sharing (PS) queue and the correlation across space is negligible in this regime. Moreover, it is easy to see that $\lambda = \lambda_c$ is an asymptote for  Equation (\ref{eqn:simulation_heuristic}) i.e. as $\lambda \rightarrow \lambda_c$, $\beta_f \rightarrow \infty$.  This further strengthens the belief that the stationary distribution is close to Poisson in the heavy-traffic regime as it predicts the correct stability boundary.  Making this claim rigorous or  even just state a mathematical conjecture is quite challenging and  would require an appropriate scaling of space and time  similar to the diffusion scaling considered for a single server PS queue \cite{gromoll_diffusion}.


\begin{figure*}
\centering
\subfloat[$\lambda = 1.40$]{
\includegraphics[scale=0.28]{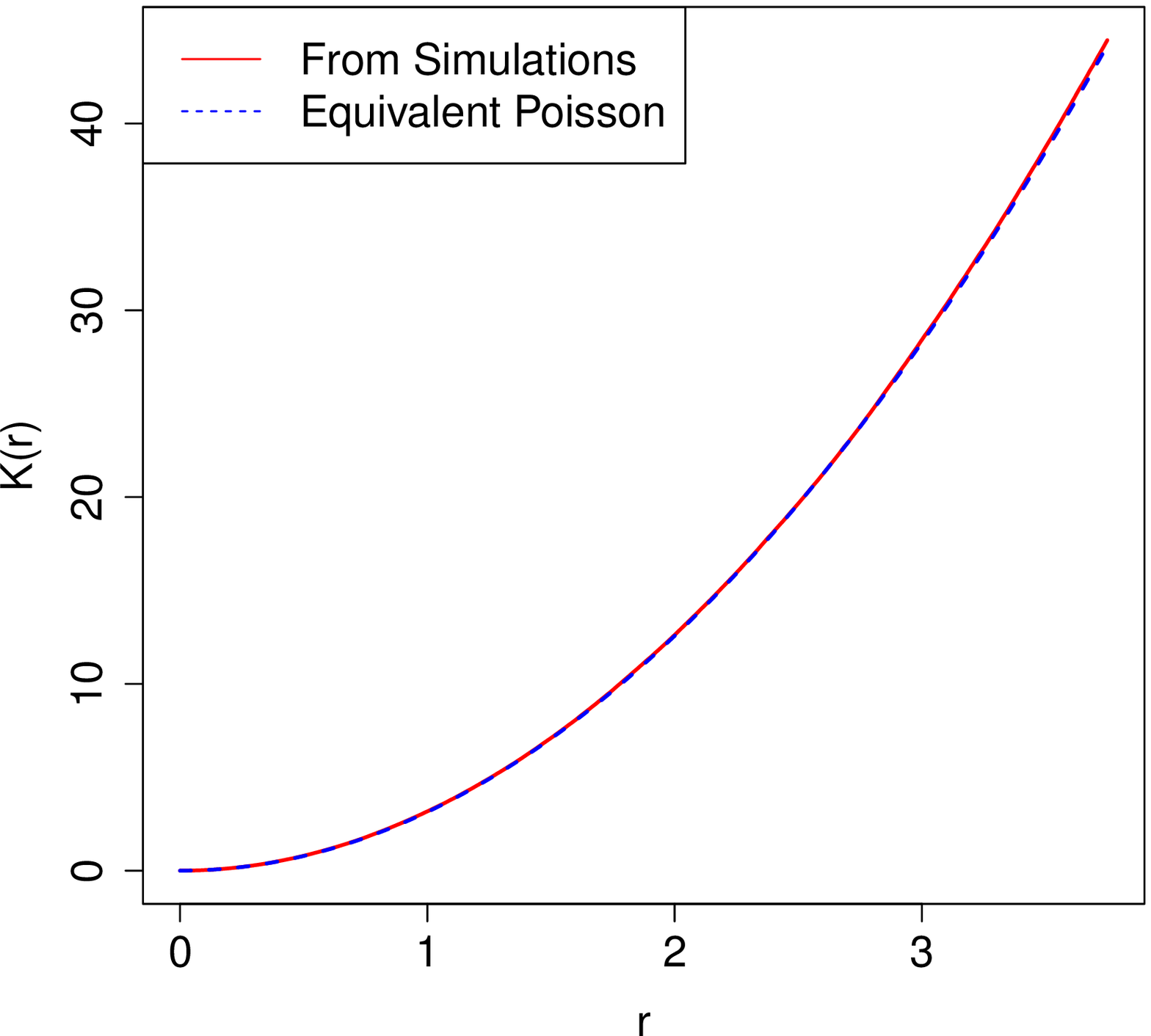}
\label{fig:clustering_1}}
\subfloat[ $\lambda = 0.99$]{
\includegraphics[scale=0.28]{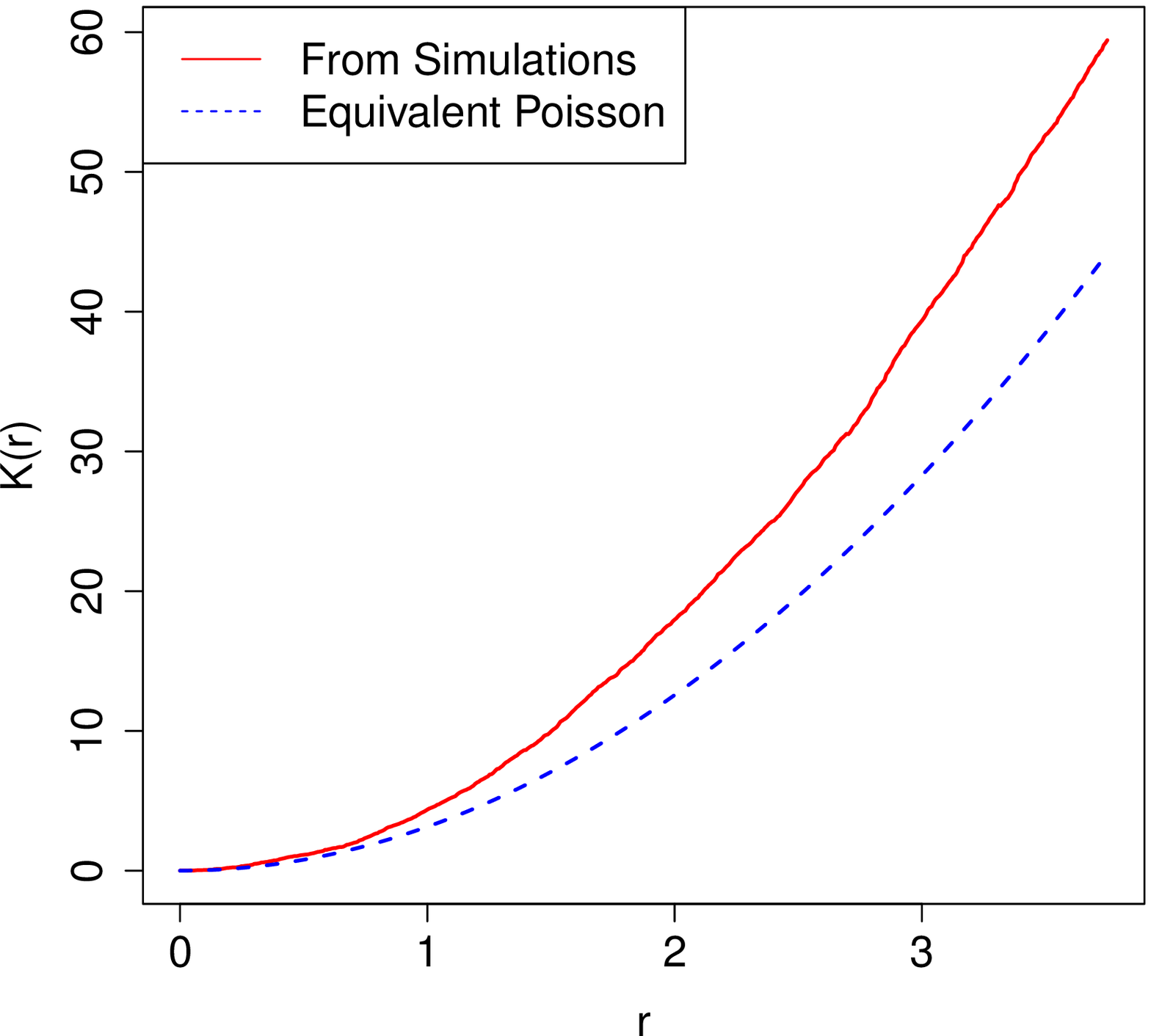}
\label{fig:clustering_2}}
\subfloat[$\lambda = 0.2$]{
\includegraphics[scale=0.28]{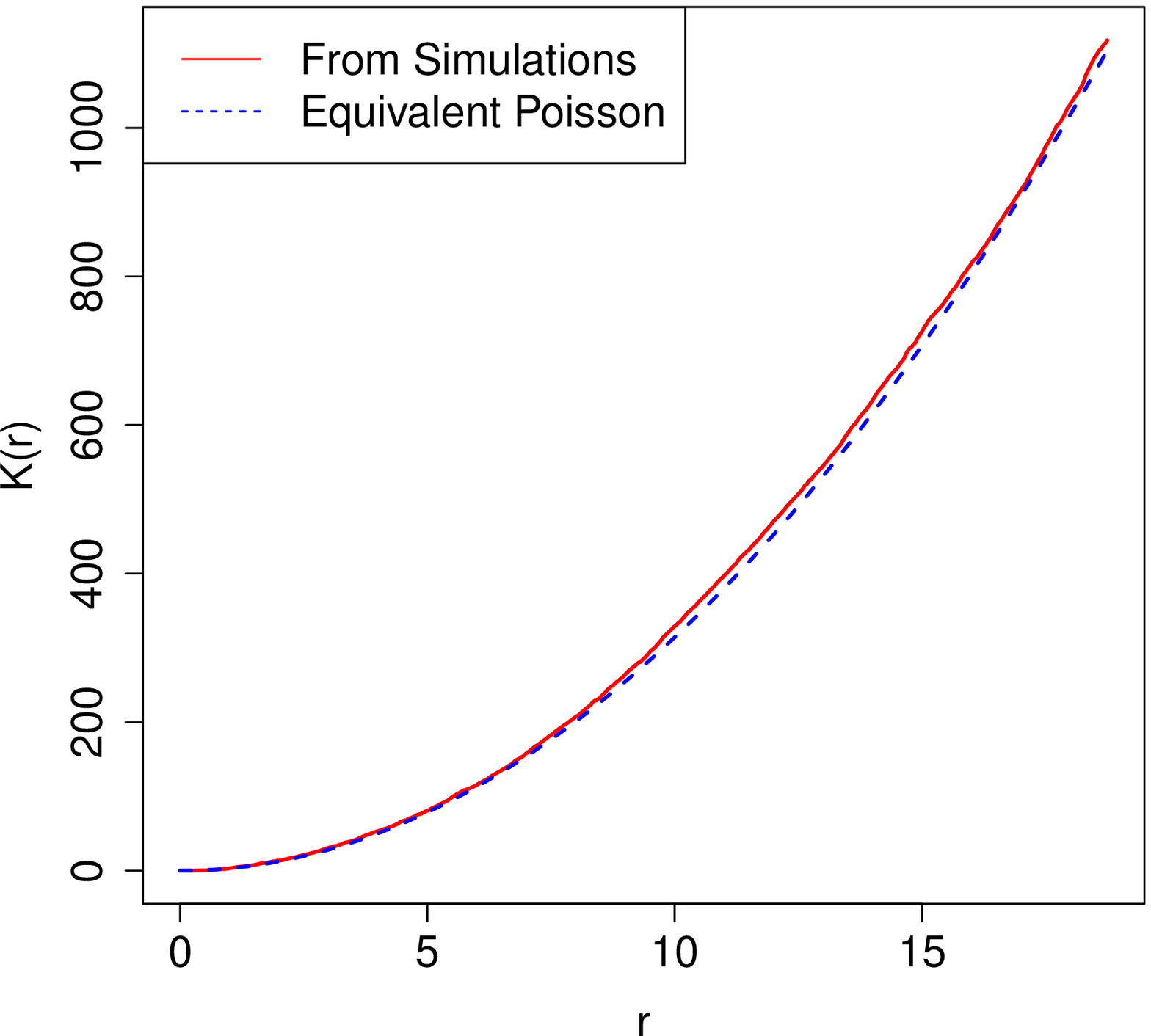}
\label{fig:clustering_3}}
\label{fig:clustering_123}
\caption{Plot comparing the Empirical Ripley K-function $K_{\phi_0}$ with that of an equivalent PPP. The path loss function is $l(r) = (r+1)^{-4}$, $T=1$. The critical $\lambda_c = 1.42$. This shows that there is little clustering in the heavy and light traffic regimes but significant clustering in the intermediate regime.}
\end{figure*}

\subsection{Tightness of $\mathbf{\beta_s}$}
We argue here that in both the low and heavy-traffic regimes, the approximation $\beta_s$ is close to $\beta_f$ and is hence a good approximation of the true $\beta$. In low-traffic, as $\lambda \rightarrow 0$, the smallest solution of Equation (\ref{eqn:fixed_point}) tends to $0$ and hence the formula for $\beta_s \approx \frac{\lambda L}{C \log_2 \left(1 + \frac{1}{N_0 }  \right)}$. This from Claim \ref{claim_low} gives that $\beta_s$ and $\beta_f$ predict the same value in low-traffic. In high-traffic regime, as $\lambda \rightarrow \lambda_c$, the value if $I_s$ from Equation (\ref{eqn:fixed_point}) is very high. Thus, the second-order moment-measure approximation $\rho^{(2)}(x,y)$ in Equation (\ref{eqn:moment_measure_approx}) is almost constant i.e. does not depend of the actual values of $x$ and $y$ as $l(\cdot)$ is a bounded function. This implies that the effect of clustering vanishes in this heuristic and hence is close to $\beta_f$.

\subsection{Intermediate Clustered Regime}

In the intermediate regime, the Poisson approximation is poor and the steady state-point process is quite clustered (see Figures \ref{fig:clustering_points} and \ref{fig:clustering_2}) i.e. $K_{\phi_0}$ is much larger than $K_{\mathrm{PPP}}$. However, we see from Figure \ref{fig:perf_comparision} that the second-order heuristic $\beta_s$ performs much better than the Poisson heuristic in this regime as it takes into account some form of spatial-correlations atleast upto second-order moment measure of $\phi_0$. However, Figure \ref{fig:clustering_points} which shows a snapshot of $\phi_0$ which is a clustered process is very interesting as it indicates finer properties of higher order moment measures. One observes for instance ``\emph{filaments}'' of points which are locally directional in-spite of the fact that the dynamics is isotropic. Such behavior indicates that the higher order moment measures of $\phi_0$ (of order greater than $2$) may  have interesting  properties which we  capture neither in Theorem \ref{thm:clustering} nor in the second order  heuristic $\beta_s$. Understanding the higher order moment measure of $\phi_0$ can also aid in proposing a \emph{provably} better performance bound in this intermediate regime. Studying these higher order moment measures of $\phi_0$ will be a very interesting and challenging direction of research.

%
%

\begin{figure}
\centering
\includegraphics[scale=0.4]{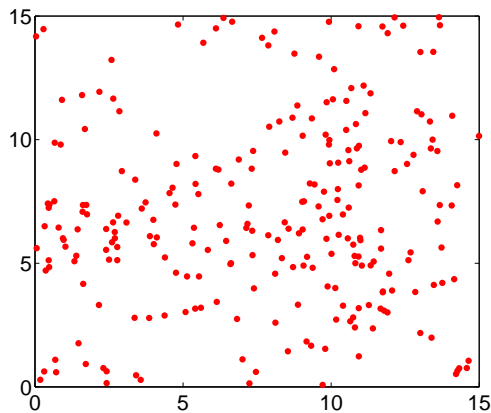}
\caption{A sample of $\phi_0$ when $\lambda = 0.99$ and $l(r) = (r+1)^{-4}$. This is a visual representation of the clustering of points.}
\label{fig:clustering_points}
\end{figure}

\subsection{Delay Tails}

To get a heuristic understanding of the delay tails, one would be tempted at first glance to approximate our model by an equivalent $M/M/1$  PS queue using a \emph{spatial-fluid} approximation that neglects randomness in space. We see through simulations that any approximation that neglects spatial interactions will predict much larger delays for a typical link than the true delays in our model. 

\begin{claim}
The delay tails in our model are exponential and have a faster decay than that of an equivalent $M/M/1$ PS queue obtained by a ``spatial-fluid'' approximation.
\end{claim}

An equivalent $M/M/1$ PS queue approximation has the following parameters - arrival rate $\lambda$, service requirement of mean $L$ and service capacity of the server $\lambda_c$ which is split equally among all customers in the queue. Such a PS queuing model is equivalent to  a first-order approximation where the spatial randomness vanishes and a point in steady-state receives rate of $C\log_2 \left(1 + \frac{1}{N_0 + \beta a} \right)$ where $\beta$ is the density of points in  steady-state. Hence, the quantity $\frac{C}{\ln(2) a}$ (which is an upper bound on the total rate given to all points i.e. $C\beta\log_2 \left(1 + \frac{1}{N_0 + \beta a} \right) \leq \frac{C}{\ln(2) a}$)  can be seen as the maximum service capacity of the spectrum in $\mathbf{S}$ which is equally shared by all links accessing the spectrum. Another simple picture as to why the above $M/M/1$ PS queue is a simple heuristic is to observe that this queue corresponds to the scenario when one ignores spatial interactions among the arriving points and assumes that the total spectrum ``capacity'' of $\lambda_c$ is shared equally among all the links sharing the spectrum in $\mathbf{S}$.  Hence, the mean-delay under the $M/M/1$ - PS model for a typical point is $\frac{L}{\lambda_c - \lambda}$ and the stability criteria for this queue is the same as that for our spatial model. However, we note from simulations (Figure \ref{fig:tails_compare_overall}) that the delay tails predicted by the heuristic $M/M/1$ queue which completely ignores spatial interactions are much larger than those observed in our model. 
\\

The poor performance of the  queuing approximation can be  understood by studying the correlation between the delays of different links. In Figure \ref{fig:delay_corr}, we plot the correlation between the delay experienced by two links arriving at the \emph{same time} as a function of their distance. We consider the $T=0$ case and hence the distance between two links is just the distance between the two points. Numerically, we plotted Figure \ref{fig:delay_corr} by first sampling a steady-state point process (by running the Markov Chain $\phi_t$ for a long time) and then introducing two additional links to this sample with independent file-sizes. We then run the dynamics from this state until the two additional links die and then compute the correlation between their delays
\\  

We see from Figure \ref{fig:delay_corr} that as the distance between the two links increases, the delays of the two links are almost uncorrelated even though they arrive at the same time. This indicates that,  two links arriving at the same time will be almost oblivious to each other and will each roughly receive independent service if they arrive far enough apart in space. This is unlike in the $M/M/1$ - PS queue approximation where  two customers arriving at the same time have positively associated delays as both of them will be competing for the same spectrum resource. This suggests that the spatial heterogeneity is key in extracting more ``service'' from the spectrum than predicted by a model which considers spectrum as a fixed quantity of good to be divided among contending links.

\begin{figure}[h]
\centering
\includegraphics[scale = 0.35]{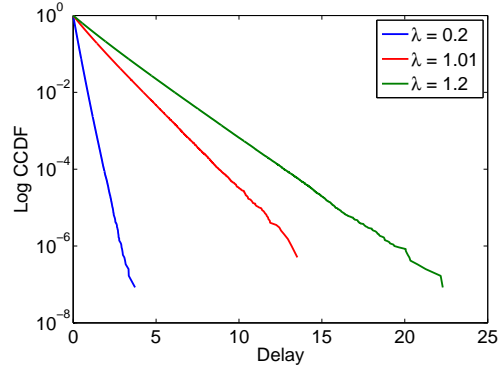}
\caption{ Plot of  logarithm of  CCDF of delay.}
\label{fig:tails_combined}
\end{figure}

\begin{figure}[h]
\centering
\includegraphics[scale=0.3]{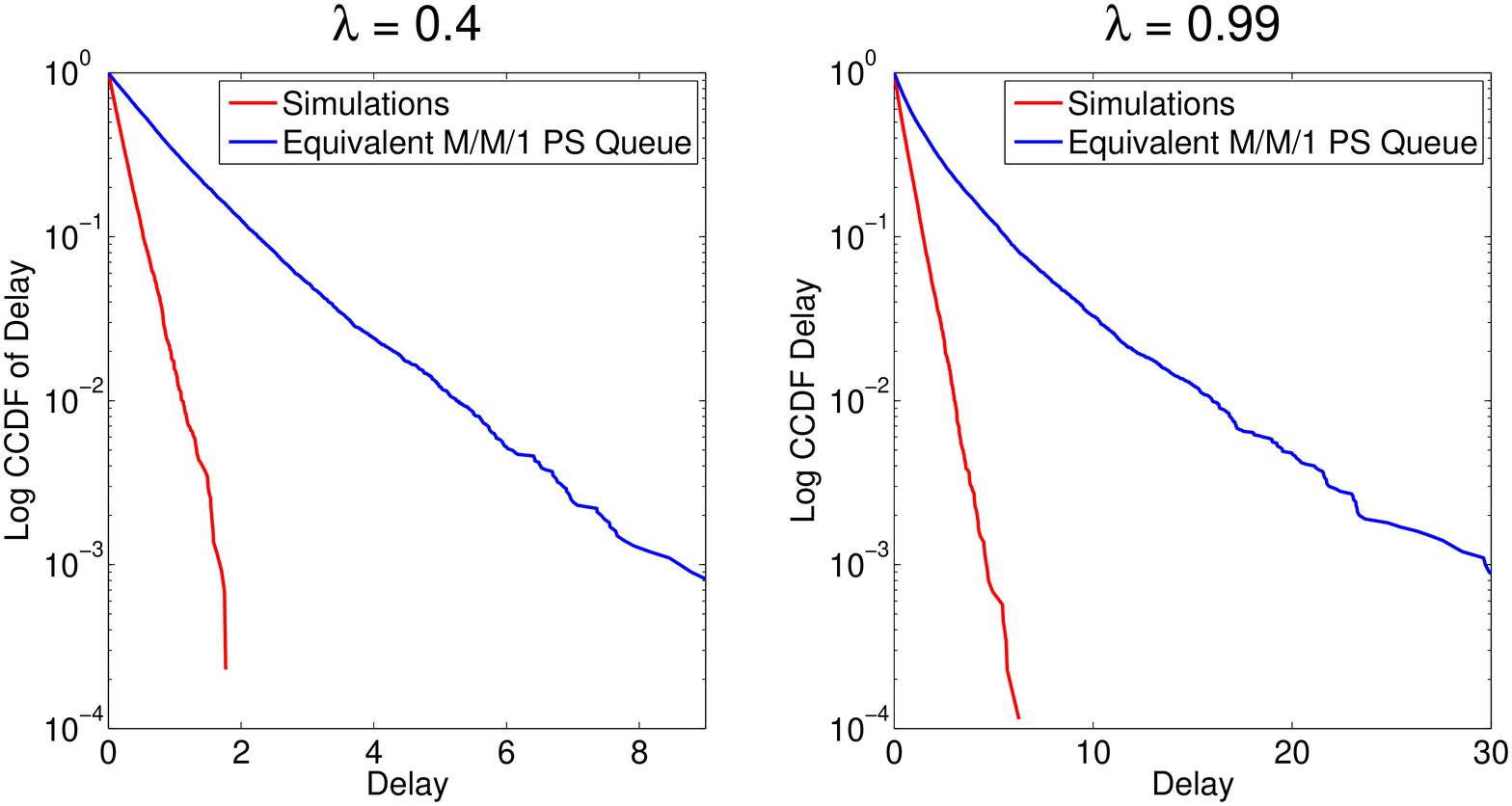}
\caption{Comparison of the delays  with that of an equivalent M/M/1 - PS queue. The critical $\lambda_c = 1.42$.}
\label{fig:tails_compare_overall}
\end{figure}

\begin{figure}
\centering
\includegraphics[scale=0.4]{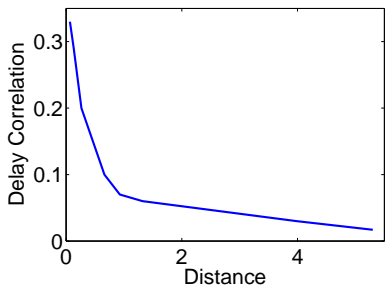}
\caption{Decay of delay correlation of two points born at the same time, as a function of their distance. $\lambda = 0.8$.}
\label{fig:delay_corr}
\end{figure}

\subsection{Heavy Tailed File Sizes}

 \begin{claim}
$\phi_t$, with file-sizes being Pareto distributed of mean $L$ and finite variance, admits a stationary regime with the critical $\lambda$ being smaller than or equal to $\frac{C}{\ln (2) L a}$. 
\label{claim_heavy}
\end{claim}

 This model also exhibits the interesting phenomenon of  prominent clustering  in the intermediate traffic regime and very little to no-clustering in the asymptotic regimes of high and low traffic. Note that the term ``high-traffic'' in this context is somewhat  loose  since we do not even know exactly the stability region. With regards to delays, our model predicts tails that are stochastically dominated by the delay of a typical customer of an equivalent $M/GI/1$ PS queue (see Figure \ref{fig:delay_heavy}). The equivalent queue we compared against had a capacity of $\lambda_c$ which from Claim \ref{claim_heavy} is an upper bound on the  capacity. Nonetheless, the delay predicted in our model is stochastically smaller. This observation again highlights the importance of taking into account the spatial heterogeneity in modeling the ``service'' provided by the spectrum.


\begin{figure}[h]
\centering
\includegraphics[scale=0.3]{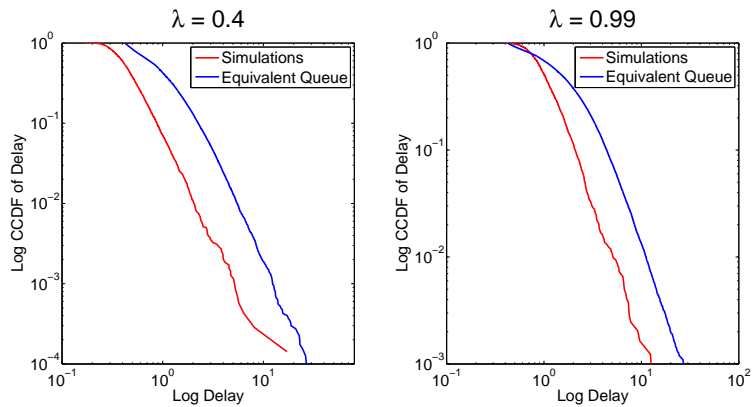}
\caption{Comparision of the delay  under Pareto file size distribution with mean $L$ and shape $\alpha=2.5$.}
\label{fig:delay_heavy}
\end{figure}

\section{Extension of the Dynamics to the Multiple Input Multiple Output (MIMO) Channel}
\label{sec:mimo_extnsions}

In the previous sections, we considered the case when both the transmitter and the receiver of a link have a single antenna. In this section, we briefly highlight, an extension of the Spatial Birth-Death (SBD) model to account for the scenario when both  transmitters and receivers have multiple antennas while still treating interference from other transmitters as noise.


\subsection{Generalized MIMO Framework for treating Interference as Noise under SBD Dynamics}

The MIMO setting is similar to the single antenna dynamics described in Section \ref{sec:system_model} except, the rate function in Equation (\ref{eqn:rate_defn}) will be modified suitably to account for the presence of multiple antennas. In this section, we provide an extension of the rate function with fading in Equation (\ref{eqn:random_rate}) to this MIMO setting. In particular, we extend Telatar's formula \cite{telatar} for MIMO channel capacity to the network case with interference treated as noise. This in itself is not new, however, when combined with our dynamic framework presents a natural example in which the framework we developed for the single antenna case naturally applies. We however note that the general case of the dynamics with MIMO is mathematically very challenging and we leave it for future work.
\\

The Telatar formula for capacity of a MIMO point-to-point channel with $X_t$ transmit antennas and $X_r$ receive antennas is given by the formula
\begin{align}
\mathcal{C} = \mathbb{E} \left[ \log_2 \left( \det(  \mathbf{I}_{X_r} + H \Sigma_{S} H^{\dag} \Sigma^{-1}_{N}   )\right)\right],
\label{eqn:telatar}
\end{align}
where the expectation is with respect to the channel matrix $H$ which is a random variable and possibly on the matrices $\Sigma_S$ and  $\Sigma_N$,  which could in general be functions of $H$. The notation $H^{\dag}$ is used to denote the complex conjugate of $H$. The matrix $\Sigma_S$ is a $X_t \times X_t$ matrix denoting the outer product of the signaling vector and $\Sigma_N$ is a $X_r \times X_r$ matrix denoting the outer product of the noise vector at the receiver. The matrix $I_{X_r}$ is the $X_r \times X_r$ identity matrix. The channel capacity formula in Equation (\ref{eqn:telatar}) captures many different scenarios such as presence or absence of Channel State Information at Transmitter (CSIT) by suitably optimizing over the correlation matrix $\Sigma_S$. For instance, in the absence of CSIT, the optimal $\Sigma_S$ is deterministic, while in the presence of CSIT, the optimal power allocation is by a water-filling on the singular values of the channel matrix $H$ and the signaling vector is along the principal components of the matrix $H$ (Chapter $8$, \cite{tse_book}). Thus, in the CSIT case, $\Sigma_S$ will be a function of $H$  and hence the expectation in Equation (\ref{eqn:telatar}) will also bear on $\Sigma_S$ in this case.
\\

We  present an extension of the formula in Equation (\ref{eqn:telatar}) to a network setting by suitably defining the rate function $R(x;\phi)$ of Equation (\ref{eqn:random_rate}).  Recall that the arrival process $\mathcal{A}$ is a marked PPP on $\mathbf{S} \times \mathbb{R}$ with the atoms denoting the receivers and the marks, which are $\mathbf{S}$ valued, denoting the location of the transmitters. We call $\phi$ a configuration of links if its atoms are locations of receivers and the marks of the atoms denote the corresponding transmitter locations. The notation $\phi^T$ is used to denote the set of transmitters or marks of the atoms of $\phi$. 

\begin{proposition}
The generalized MIMO rate function $R(x;\phi)$ where $\phi$ is a configuration of links on $\mathbf{S}$ and $x \in \phi$ is an atom of the point-process $\phi$ is given by 
\begin{align}
R(x;\phi) = C\mathbb{E} \left[ \log_2 \left( \det(  \mathbf{I}_{X_r} + H_{xx} S_{x} S_{x}^{\dag} H^{\dag}_{xx} \Sigma^{-1}_{N} (\phi; \{H_{yx}\}_{y \in \phi \setminus \{x\}}, \{S_y\}_{y \in \phi \setminus \{x\}})   )\right)\right],
\label{eqn:mimo_rate}
\end{align}
where the expectation is with respect to the i.i.d. collection of channel random matrices $\{H_{xy}\}_{y \in \phi}$ and the i.i.d. collection of signaling vectors $\{S_y\}_{y \in \phi}$. $\mathbf{I}_{X_r}$ denotes a $X_r \times X_r$ identity matrix. The $X_t \times 1$ vector $S_x$ denotes the signaling vector  of the transmitter located at $u$ whose corresponding receiver is at location $x$. The vector $S_x$ may or may not depend on $H_{xx}$ depending on whether CSIT is present or absent, but nonetheless these vectors are i.i.d. across $x$. The matrix $\Sigma_{N} $ which is a $X_r \times X_r$  matrix,  is the outer product of the   interference vector plus noise, i.e. $\Sigma_N = N_0 \mathbf{I}_{X_r} + (\sum_{y \in (\phi^T \setminus \{u\} )} \sqrt{l(||x-y||)}H_{yx} S_y) (\sum_{y \in (\phi^T \setminus \{u\} )} \sqrt{l(||x-y||)}H_{yx} S_y)^{\dag}$.
\end{proposition}
\begin{proof}

For a static and deterministic configuration of links $\phi$, and for a receiver at location $x \in \phi$ with its corresponding transmitter at location $u$ (referred to as the tagged link in  this proof), we need to argue that Equation (\ref{eqn:mimo_rate}) is the capacity of this link under fast-fading when treating interference as noise.  Assume that the channel between any transmitter whose receiver is at location $a \in \phi$ and any receiver $b \in \phi$ is given by $H_{ab}$ and is i.i.d. across $a$ and $b$ and equal in distribution to a random matrix $H$. Denote by $S_x$ as the $X_t \times 1$ random signaling vector of the transmitter of the tagged link. Note that $S_x$ could possibly depend on the channel realization $H_{xx}$ depending on whether there is CSIT or not. The interference signal at the receiver in location $x$ is  
 \begin{align}
 \mathcal{I}(x;\phi) = \sum_{y \in (\phi \setminus \{x\} )^{Tx} } \sqrt{l(||x-y||)}H_{yx} S_y,
 \end{align}
 where $H_{yx}$ and $S_y$ are i.i.d. and independent of each other with $H_{yx}$ equal in distribution to $H$ and $S_y$ equal in distribution to $S_x$. Thus, the matrix $\Sigma_{N}$ is the sum of the outer product of $\mathcal{I}(x;\phi)$  and the thermal noise co-variance matrix ${N_0} \mathbf{I}_{X_r} $. Recall that the path-loss function $l(\cdot)$ denotes the attenuation in the signal power and hence, the signal itself is attenuated by $\sqrt{l(\cdot)}$. Now using the Telatar formula of Equation (\ref{eqn:telatar}) for the case when the noise signal is the sum of thermal noise and interference, we get Equation (\ref{eqn:mimo_rate}).
\end{proof}


This propositions is ofcourse not new, but we include it here for completeness. This formulation allows us to define the birth-death dynamics in the MIMO setting. The dynamics can be described through Equation (\ref{eqn:dynamics_defn}) with the rate-function as given in Equation (\ref{eqn:mimo_rate}). The formulation in Equation (\ref{eqn:mimo_rate}) is the network version of channel capacity of MIMO under presence of fast fading and treating interference as noise. This generalized setup of the MIMO channel also allows us to study the Multiple Input Single Output (MISO) and Single Input Multiple Output (SIMO) cases by setting $X_r$ or $X_t$ to $1$ respectively. In the sequel, we discuss an example of the MIMO framework which can be analyzed as a corollary of the single antenna system.

\subsection{Independent Channels with no Channel State Information at Transmitter (CSIT)}
\label{sec:mimo_special}

We show in this sub-section, that in the special case when the signal $S_x$ is Gaussian with outer product equal to $\frac{1}{X_t} \mathbf{I}$ and the channel matrices are equal in distribution to a random matrix $H$ where each entry is i.i.d. complex normal with $0$ mean and unit variance, then the MIMO dynamics can be reduced to an equivalent single-antenna system and the critical arrival rate for this model can then be computed. Note that the transmission strategy where $S_{x}S^{\dag}_{x} = \frac{1}{X_t} \mathbf{I}$  is optimal in the case  when there is no CSIT, total transmit power constraint of $1$ and the channel matrix is composed of i.i.d. entries (Chapter $8$, \cite{tse_book}).
\\

The following statistical assumptions model the independent channel MIMO system for which the critical arrival density $\lambda_c$ can be computed as a corollary of the single antenna analysis.
\begin{itemize}
\item All channel realizations between any transmit antenna and receive antenna are i.i.d. complex normal with $0$ mean and unit variance. 
\item For any coordinate $i \in [1,X_t]$, and any link $x$ and at any time $t$, we have $\mathbb{E}[(S_x^t)_i (S_x^t)_{i}^{*}] = 1/X_t$ where the vector $S_x^t$ is the transmitted signal by transmitter whose receiver is at location $x$ at time $t$. This indicates that the total power $1$ is split equally on each antenna.
\item For any coordinates $i,j \in[1,X_r]$ and any two links $x \neq y$ and any time $t$, we have $\mathbb{E}[(S_x^t)_i (S_y^t)_{j}^{*}] = 0$. This assumption gives that the signal across antennas are uncorrelated. 
\end{itemize}

Under the foregoing assumptions, the rate function in Equation (\ref{eqn:mimo_rate}) can be simplified to 

\begin{align}
R(x,\phi) = C\mathbb{E}_{h} \left[ \sum_{i=1}^{X_r} \log_2 \left(  1 + \frac{1 }{ X_t(N_0 + I(x;\phi))} \sigma_i(H H^\dag) \right) \right],
\label{eqn:mimo_rate_defn_indep}
\end{align} 
where $H$ is a $X_t \times X_r$ random matrix denoting the channel statistics.  The quantity $\sigma_i(A)$ refers to the $i$th eigen-value of the matrix $A$ where the eigenvalues are indexed in some arbitrary fashion. The interference $I(x,\phi)$ is just a scalar and is given by 
\begin{align}
I(x,\phi) = \sum_{y \in \phi^T \setminus \{u\}} l(||y-x||),
\end{align}
where $\phi^T$ is the set of points on $\mathbf{S}$ corresponding to the transmitters and the transmitter of the receiver at location $x \in \phi$ is assumed to be present at $u \in \phi^T$. If one, employs the MIMO rate Equation (\ref{eqn:mimo_rate_defn_indep}), then one gets the following result:
\begin{corollary}
The critical arrival intensity of links $\lambda_c$ under the foregoing assumptions is $\frac{C X_r}{La \ln(2)}$. 
\label{cor:mimo_result}
\end{corollary}
We provide a proof sketch in Appendix \ref{sec:mimo_results}. We phrase the above result as a corollary since it is not surprising to have as critical density in this case of independent MIMO channel with no CSIT as $X_r$ times the critical density for a single-antenna link based SBD process. The total transmit power is $1$ just as in the case with single antenna, however the presence of $X_r$ receive antennas per link implies that the network can support upto $X_r$ times more transmitters than in the case with single antenna link.
This result indicates that the effect of having multiple antennas at the transmitter is not beneficial if there is no CSIT. On the other hand, in the presence of CSIT, one would expect to receive gain from the presence of multiple transmit antennas as $S_x$ is a function of the channel realization $H_{xx}$ thereby exploiting the diversity from multiple transmit antennas better. However, we do not pursue this question in the present paper and leave the analysis of the generalized MIMO system to future work.

\section{Conclusion and Future Work}

In this paper, we proposed a novel space-time interacting particle system to model spectrum sharing in ad-hoc wireless networks. We computed exactly the phase-transition point for time ergodicity. We also proved the intuitive fact that the steady-state point-process corresponding to this dynamics exhibits clustering. In order to understand the performance metric of density of links in steady-state, we proposed a Poisson heuristic $\beta_f$ (which is a bound subject to Conjecture \ref{conjecture_1}) and a second order heuristic $\beta_s$. We saw from simulations that both the heuristics are tight in the two asymptotic regimes of heavy and light traffic. However, in the intermediate traffic regime, we found that the heuristic $\beta_s$ performs much better compared to the Poisson heuristic $\beta_f$ as $\beta_s$ accounts for some spatial correlations which are non negligible in this regime. We also saw through simulations that any form of simplistic modeling of spatio-temporal interactions through PPP or equivalent queues ignoring spatial clustering, leads to poor estimates for  performance.
\\

From a mathematical perspective, we identified several challenging directions of future work in the simulation section. In particular, understanding the higher order moment measure of $\phi_0$ will  be key in evaluating or providing provably tighter bounds for performance metrics. Understanding the higher-order moment measures may also aid in making progress on Conjecture \ref{conjecture_1}. From an information-theoretic perspective, we considered a dynamic interference network where links treat interference as noise. However, it will be interesting to consider other receiver schemes such as Successive Interference Cancellation or Joint-Decoding and show that the critical arrival rate for these schemes are strictly better than considering all Interference as Noise. This will then yield the complete dynamic version of the model considered in \cite{baccelli_gammal}, namely a dynamic version of an interference network with point-to-point codes.



\section*{Acknowledgements}

This work was supported by an award from the Simons Foundation ($\# 197982$) to The University of Texas at Austin and from the grant
No. NSF-CCF-1218338. The authors also acknowledge the support of TACC (Texas Advanced Computing Center) for providing access to computing resources to perform the simulations.

\bibliographystyle{IEEEtran}
\bibliography{sigproc}


\newpage

\section*{Appendix}

\section{Proof of Theorem \ref{thm:necessary_condn}}
\label{sec:proofs}

\label{sec:proof1}
\begin{proof}
We prove this by contradiction. Assume that $\phi_t$ is in stationary regime and that $\lambda > \frac{C l(T)}{\ln (2) L a}$. We use the Miyazawa's Rate-Conservation Principle or Law (RCL) (e.g. \cite{queuing_theory}, 1.3.3)  to set-up a system of equations and identify a contradiction. Applying the RCL to the stochastic process $\phi_t(\mathbf{S})$ which counts the number of links yields,
\begin{equation}
\lambda |\mathbf{S}| = \lambda_d,
\label{eqn:proof1_RCL_Number}
\end{equation}
where $\lambda_d$ is the intensity of the point-process on $\mathbb{R}$ corresponding to the epochs of a death-time. Since we assumed that $\phi_t$ is in stationary regime, the point process formed on the real line by the instants of a death is stationary with intensity $\lambda_d = \lambda |\mathbf{S}|$.  Applying  RCL to the total ``work-load'' in the network i.e. the total number of bits that each of the transmitters present are yet to send to their corresponding receivers, we get
\begin{equation}
\lambda|S|L = \mathbb{E} \left[ \sum_{x \in \phi_0} R(x, \phi_0) \right],
\label{eqn:proof1_RCL_bits1}
\end{equation}
where $R(x,\phi)$ is given in Equation (\ref{eqn:rate_defn}). From the definition of Palm Probability of $\phi_0$, we have that 
\begin{equation}
\lambda |\mathbf{S}| L = \mathbb{E}^{0}_{\phi_0} \left[R(0,\phi_0) \right] \mathbb{E}[\phi_0(\mathbf{S})],
\label{eqn:proof1_RCL_bits2}
\end{equation}
where $\mathbb{E}^{0}_{\phi_0}$ is the (spatial) Palm Probability of $\phi_0$ and ${\phi_0(\mathbf{S})}$ is the random variable denoting the number of links in the network in steady-state. Note that from our assumption that $\phi_t$ is in stationary regime ensures the existence of the Palm Probability measure of the spatial point process $\phi_0$. Applying rate-conservation to the stochastic process  $\mathbf{I}_t = \sum_{x \in \phi_t} I(x,\phi_t)$, the sum interference seen at all receivers (which could possibly be $\infty$), we get
\begin{equation}
\lambda |\mathbf{S}| \mathbb{E}^{\uparrow}[\mathcal{I}] = \lambda_d \mathbb{E}^{\downarrow}[\mathcal{D}],
\label{eqn:proof1_RCL_interference1}
\end{equation}
with $\mathcal{I} = \mathbf{I}_{0_{+}} - \mathbf{I}_{0}$ and $\mathcal{D} = \mathbf{I}_{0} - \mathbf{I}_{0_{+}}$. Here, $\mathbb{E}^{\uparrow}$ denotes the (time) Palm probability corresponding to the point process on $\mathbb{R}$ of birth instants and $\mathbb{E}^{\downarrow}$ denotes the (time) Palm probability of the point process on $\mathbb{R}$ corresponding to the instants of death. From Equation (\ref{eqn:proof1_RCL_Number}) we have 
\begin{equation}
\mathbb{E}^{\uparrow}[\mathcal{I}]  = \mathbb{E}^{\downarrow}[\mathcal{D}].
\label{eqn:proof1_RCL_interference2}
\end{equation}
From the PASTA  property and the fact that the births are uniform in $\mathbf{S}$, we have from Campbell's theorem that
\begin{equation}
\mathbb{E}^{\uparrow}[\mathcal{I}] = 2 \mathbb{E}[ \phi_0(\mathbf{S})] \frac{a}{|\mathbf{S}|}.
\label{eqn:proof1_RCL_interference3}
\end{equation}

Since the file-sizes at all transmitters are i.i.d. exponential with mean $L$, the point process on the real line corresponding to the death-instants admits as stochastic-intensity $\mathbf{R}_t = \frac{1}{L} \sum_{x \in \phi_t} R(x,\phi_t)$ with respect to the filtration $ \mathcal{F}_t = \sigma ( \phi_s : s \leq t)$, the sigma algebra corresponding to the locations. Hence, it then follows from Papangelou's theorem (e.g. \cite{queuing_theory}, Theorem $1.9.2$) that 
\begin{equation}
\frac{d \mathbb{P}^{\downarrow}}{d \mathbb{P}} \vert_{\mathcal{F}_{0_{-}}} = \frac{\mathbf{R}_0}{\mathbb{E}[\mathbf{R}_0]}.
\label{eqn:proof1_papangelou}
\end{equation}
Since the decrease in total interference (in state $\phi_{0_{-}}$) is of magnitude $I(X,\phi_{0})$ with probability $\frac{R(X,\phi_0)}{L\mathbf{R}_0}$ if $X \in \phi_{0_{-}}$, we get
\begin{align}
\mathbb{E}^{\downarrow}[\mathcal{D}] &= 2 \mathbb{E} \left[ \frac{\mathbf{R}_0}{\mathbb{E}[\mathbf{R}_0]} \sum_{x \in \phi_0} \frac{R(x,\phi_0)}{L \mathbf{R}_0} I(x,\phi_0) \right] \nonumber \\
&= 2 \frac{\mathbb{E}[\sum_{x \in \phi_0} R(x, \phi_0) I(x,\phi_0)]}{L \mathbb{E}[\mathbf{R}_0]} \nonumber \\
&= 2 \frac{\mathbb{E}^{0}_{\phi_0}[ R(0, \phi_0) I(0,\phi_0)]}{L \mathbb{E}[\mathbf{R}_0]}  \mathbb{E}[\phi_0(\mathbf{S})].
\label{eqn:proof1_RCL_interference4}
\end{align}
Now combining, Equations (\ref{eqn:proof1_RCL_interference4}), (\ref{eqn:proof1_RCL_interference3}) and (\ref{eqn:proof1_RCL_bits1}), we get
\begin{equation}
a = \frac{\mathbb{E}^{0}_{\phi_0}[ R(0, \phi_0) I(0,\phi_0)]}{L \lambda}.
\label{eqn:proof1_penultimate}
\end{equation}

From Equation (\ref{eqn:rate_defn}) and basic calculus, we have that  $R(0, \phi_0) I(0,\phi_0) \leq \frac{C l(T)}{\ln (2) }$ which is a deterministic bound that is true for \emph{any} $\phi \in \mathbf{M}(\mathbf{S})$. Applying this inequality to Equation (\ref{eqn:proof1_penultimate}), we get the inequality that
\begin{equation}
\lambda \leq \frac{C l(T)}{\ln(2) L a}.
\label{eqn:proof1_final}
\end{equation}
Inequality (\ref{eqn:proof1_final}) is a contradiction to our assumption that $\phi_t$ is in stationary regime and that $\lambda >   \frac{C l(T)}{\ln(2) L a}$. 
\end{proof}

\section{Proof of Theorem \ref{thm:sufficient_condn}}

 For simplicity of the proof,  we assume the link distance $T = 0$. The proof for arbitrary $T$ follows with significantly more notation that obscures the essence of the proof. Thus, to keep the proof ideas simple, we first outline the proof for the special case of $T=0$ with remarks in between as to how the intermediate steps can generalize. At the end, we will give the complete construction of the coupling (which will be explained later) in the general case of $T$ being arbitrary. This will then complete the proof in the general case as well.
 \\
 
 Assume $T=0$ for the time being. Thus, the dynamics is that of points arriving and exiting the network and the network at any point of time consists of a collection of points distributed in space $\mathbf{S}$. The high level idea of the proof is that we tesselate the space $\mathbf{S}$ and study another ``upper-bound'' Markov Chain living on a countable state-space which we analyze through fluid limit techniques. We then conclude about the ergodicity of $\phi_t$ which is a Markov Chain on the topological space $\mathbf{M}(\mathbf{S})$. 
\\

To define the upper-bound chain, we first tessellate the square $\mathbf{S}$ into cells where each cell is a square of length exactly $\epsilon$. Since $\mathbf{S}$ is a torus, we  assume without loss of generality that the origin is in the center of a cell. One can find a sequence of such  tessellations with  the side length of the cells going to $0$.  The tessellation for each valid $\epsilon > 0$ results in $n_{\epsilon}$, a finite number of cells as $\mathbf{S}$ is compact. Index the cells by $i$ and let $A_i$ denote the subset of $\mathbf{S}$ corresponding to cell $i$ and $a_i \in A_i$ denote its center. The cell containing the origin is indexed $0$ i.e. $a_0 = 0$.  For such an $\epsilon$ tessellation, we define a new path-loss function $l_{\epsilon}(x,y)$ where $l_{\epsilon}(x,y) = l_{\epsilon}(a_i,a_j)$ for all $x \in A_i$ and $y \in A_j$ and
\begin{align}
l_{\epsilon}(a_i,a_j) = \sup \{ l(|| b_i - b_j ||) : ||b_i - a_i || , || b_j - a_j || \in \{ 0 , \epsilon \} \}. \nonumber
\end{align}

Note that the function $l_{\epsilon}$ satisfies  
\begin{equation}
\sum_{i}l_{\epsilon}(a_i,a_j) = \sum_{i}l_{\epsilon}(a_i,0) = \frac{1}{\epsilon^2}\int_{x \in \mathbf{S}}l_{\epsilon}(||x||)dx,
\label{eqn:proof3_torus_symmetry}
\end{equation}
since $\mathbf{S}$ is a square torus and each cell $A_i$ is a square of side-length  $\epsilon$,

The upper bound Markov-Chain is denoted as $\phi_{t}^{(\epsilon)}$ which takes value in the space $\mathbf{M}(\mathbf{S})$. This chain has the exact same dynamics as described in Equation (\ref{eqn:dynamics_defn}) \emph{except} that the interference comes from $l_{\epsilon}(.,.)$ instead of from $l(\cdot)$,

\begin{lemma}
For all time $t$, the point-process $\phi_{t}^{(\epsilon)}$ stochastically dominates $\phi_t$. This implies that if $\phi_{t}^{(\epsilon)}$ is stable for a particular $\lambda$, then so is $\phi_t$ for that value of $\lambda$.
\end{lemma}

\begin{proof}

We have from the monotonicity of $l(\cdot)$, $l_{\epsilon}(x,y) \geq l(x,y) = l(||x-y||)$ for each $x ,y \in \mathbf{S}$. Thus, for each $x \in \mathbf{S}$ and each $\phi \in \mathbf{M}(\mathbf{S})$, $I_{\epsilon}(x,\phi) \geq I(x,\phi)$ and subsequently $R_{\epsilon}(x,\phi) \leq R(x,\phi)$ as $R(x,\phi)$ is a decreasing function of $I(x,\phi)$. Therefore the point process $\phi_{t}^{(\epsilon)}$ \emph{stochastically dominates} the point process $\phi_t$. This follows from the fact that for any $\phi \in \mathbf{M}(\mathbf{S})$, we have that the birth rate $\lambda |S|$ is the same for both process, whereas the death-rate of each point of $x \in \phi$ satisfies $\frac{1}{L}R_{\epsilon}(x,\phi) \leq \frac{1}{L} R(x,\phi)$. Also form Equation (\ref{eqn:rate_defn}), if $\phi_1 \subseteq \phi_2$, then for each $x \in \phi_1 \cap \phi_2$, $R(x,\phi_1) \geq R(x,\phi_2)$.  Hence one can construct a coupling of the process $\phi_t$ and $\phi_{t}^{(\epsilon)}$ such that $\phi_t \subseteq \phi_{t}^{(\epsilon)}$ $\forall t$, i.e.  a point is alive in $\phi_t$ only if it is also alive in $\phi_{t}^{(\epsilon)}$. Therefore,  if  $\phi_{t}^{(\epsilon)}$ is ergodic for a given $\lambda$, then $\phi_t$ is also ergodic for that arrival rate  $\lambda$. 
\end{proof}

Define $\mathbf{X}^{(\epsilon)}(t) = \{\phi_{t}^{(\epsilon)}(A_i)\}_{i=1}^{n_{\epsilon}}$  as the $n_{\epsilon}$ dimensional vector taking values in $\mathbb{N}^{n_{\epsilon}}$. It is easy to see that  $\mathbf{X}^{(\epsilon)}(t)$ is a Markov-Chain since the path-loss function $l_{\epsilon}(x,y)$ does not distinguish between two different locations of space inside a cell. It is also evident that if $\mathbf{X}^{(\epsilon)}(t)$ is ergodic, then $\phi_{t}^{(\epsilon)}$ is ergodic since $\lim_{t \rightarrow \infty}\mathbb{P}[ \phi_{t}^{(\epsilon)}(\mathbf{S}) < \infty] = \lim_{t \rightarrow \infty} \mathbb{P} [||\mathbf{X}^{(\epsilon)}(t)||_{1} < \infty] = 1$. The second equality follows from the fact that $\mathbf{X}^{(\epsilon)}(t)$ is a finite-dimensional ergodic Markov chain on $\mathbb{N}^{n_\epsilon}$. Hence, a sufficient condition for stability of $\phi_t$ is a condition for the Markov Chain $\mathbf{X}^{(\epsilon)}(t)$ to be ergodic.
\\


We  show in Theorem \ref{thm:XofT} that $\mathbf{X}^{(\epsilon)}(t)$ (and hence $\phi^{(\epsilon)}_t$) is ergodic if 
\begin{equation}
\lambda < \frac{C}{L \ln(2) \int_{x \in \mathbf{S}}l^{\epsilon}(x,0)dx},
\label{eqn:proof3_epsilon_stability_condn1}
\end{equation}
which will actually conclude the proof of Theorem \ref{thm:sufficient_condn}. This can be seen as follows. Since the point process $\phi^{\epsilon}_t$ stochastically dominates $\phi_{t}$, we can optimize the stability region in Equation (\ref{eqn:proof3_epsilon_stability_condn1}) by choosing the best $\epsilon$.  As the function $r \rightarrow l(r)$ is monotone, $l_{\epsilon}(x,0)$ is monotone increasing in $\epsilon$ for each $x \in \mathbf{S}$ and hence we want to have $\epsilon$ as small as possible. Furthermore, the function $r \rightarrow l(r)$ has only a countable set of discontinuity points (as it is bounded non-increasing), we have that as $\epsilon$ goes to $0$, $l_{\epsilon}(x,0)$ converges to $l(x,0)$ for almost-every $x \in \mathbf{S}$. Hence, $\lim_{\epsilon \rightarrow 0}\int_{x \in \mathbf{S}}l_{\epsilon}(||x||)dx = \int_{x \in \mathbf{S}}l(||x||)dx$ from the Monotone Convergence theorem.  Therefore, if  $\mathbf{X}^{(\epsilon)}(t)$ is ergodic under condition in Equation (\ref{eqn:proof3_epsilon_stability_condn1}), then $\phi_t$ will be ergodic under the condition

\begin{align}
 \lambda &< \lim \sup_{\epsilon \rightarrow 0}\frac{C}{L \ln(2) \int_{x \in \mathbf{S}}l^{\epsilon}(x,0)dx} = \frac{C}{L \ln(2)\int_{x \in \mathbf{S}}l(x,0)dx},
 \end{align}
 which will conclude the proof of Theorem \ref{thm:sufficient_condn}.


\begin{theorem}
$\mathbf{X}^{(\epsilon)}(t)$ is ergodic under the condition in Equation (\ref{eqn:proof3_epsilon_stability_condn1}).
\label{thm:XofT}
\end{theorem}
We remark that, even in the general case of $T > 0$, the same theorem statement will hold for a slightly modified version of $\mathbf{X}^{(\epsilon)}(t)$ which we will construct later. Thus, if Theorem \ref{thm:XofT} is established for arbitrary $T$, the proof of the main theorem will be complete by similar reasoning in the previous paragraph.

\begin{proof}

We can write the following evolution for the vector $\mathbf{X}^{(\epsilon)}(t)$ which we refer to as $\mathbf{X}(t)$ in the sequel for convenience as
\begin{align*}
X_i &\rightarrow X_i + 1 \text{  at rate } \lambda \epsilon^2 \\
X_i &\rightarrow X_i - 1 \text{  at rate } \\& X_i \log_2 \left( 1 + \frac{1}{N_0 + \sum_{j=1}^{n_{\epsilon}} (X_j - \mathbf{1}(j=i)) l_{\epsilon}(a_i,a_j)} \right).
\end{align*}

We note that generalizing this dynamics to the case when $T > 0$ is slightly different alebit the same principles and we outline it at the end of the proof.
\\

Under condition in Equation (\ref{eqn:proof3_epsilon_stability_condn1}), we show the following drift argument to hold which will conclude the proof.
\begin{theorem} \cite{Roberts_Book} 
Let $\mathbf{X}(t)$ be a Markov Chain taking values in a countable state space $\mathcal{S}$. Assume there exists a function $L : \mathcal{S} \rightarrow \mathbb{R}_{+}$ and constants $A < \infty$ , $\epsilon > 0$ and an integrable stopping time $\hat{\tau} > 0$ such that for all $x \in \mathcal{S}$:
\begin{equation}
L(x) > A \implies \mathbb{E}_{x}L(\mathbf{X}(\hat{\tau})) \leq L(x) - \epsilon \mathbb{E}_{x}(\hat{\tau}).
\label{eqn:proof3_markov_lyapunov_thm}
\end{equation}
If in addition the set $\{x: L(x) \leq A \}$ is finite and \\ $\mathbb{E}_{x}L(\mathbf{X}(1)) < \infty$ for all $x \in \mathcal{S}$, then  $\mathbf{X}(t)$ is ergodic.
\end{theorem} 

We will show that the above theorem is satisfied with the Lyapunov function $L(x) = ||x||_{\infty}$ and 
\begin{equation}
\hat{\tau} = L(\mathbf{X}(0))\left(  \frac{C}{L \ln(2)  \sum_{k=0}^{n_{\epsilon}-1}l_{\epsilon}(a_k,0)} - \lambda \epsilon^2 \right)^{-1} := L(\mathbf{X}(0))\tau,
\label{eqn:proof3_tau_defn}
\end{equation}
 a deterministic finite stopping-time. We will use the notation that $||x||_{\infty} = |x|$ which is also equal to $L(x)$.

To establish the drift condition, we pass to the fluid-limit. A fluid limit of the Markov-Process $\mathbf{X}(t)$ is denoted by $x(t)$ which is a $n_{\epsilon}$ dimensional vector. $x(t)$ is defined as a fluid limit if there exists non-decreasing Lipschitz continuous function $\{D_{i}(t)\}_{i=0}^{n_{\epsilon}-1}$ such that 
\begin{align*}
x_{i}(t) = x_{i}(0) + \lambda \epsilon^2 t - D_{i}(t),
\end{align*} 
where the derivative of $D_{i}(t)$ satisfies  $\dot{D}_{i}(t) = \frac{Cx_{i}(t)}{L \ln(2)\sum_{j=0}^{n_{\epsilon}-1} x_j(t) l_{\epsilon}(a_i,a_j)}$, or equivalently, the fluid limit $x(t)$ satisfies the following set of differential equations. If $||x(t)||_{\infty} > 0$,
\begin{equation}
\frac{d}{dt}x_{i}(t) = \lambda \epsilon^2 - \frac{C x_{i}(t)}{L \ln(2) \sum_{k=0}^{n_{\epsilon}-1} x_{k}(t) l_{\epsilon}(a_i,a_k)}
\label{eqn:proof3_diff_eqn_defn}
\end{equation}
and if $||x(t)||_{\infty} = 0$, 
\begin{equation*}
\frac{d}{dt}x_{i}(t) = 0.
\end{equation*}

For $y \in \mathbb{R}^{n_{\epsilon}}$, denote by $S(y)$  the set of fluid functions $x(t)$ such that $x(0) = y$. The following theorem establishes that the above fluid equation is indeed obtained through an appropriate space and time scaling. It also establishes as a corollary that $S(y)$ is non-empty for any $y \in \mathbb{R}^{n_{\epsilon}}$.
\begin{theorem}
Consider a sequence of deterministic initial conditions    $\{{X}^{(k)}(0)\}_{k \geq 1}$ for the Markov Chain $\mathbf{X}(t)$ and a sequence of positive integers $\{z_k\}_{k \geq 1}$ with $\lim_{k \rightarrow \infty}z_k = \infty$ such that the limit $\lim_{k \rightarrow \infty}z_{k}^{-1}X^{(k)}(0) = x(0)$ exists. Then for all $s>0$ and all $\delta > 0$, the following convergence takes place
\begin{equation}
\lim_{k \rightarrow \infty} \mathbb{P} \left( \inf_{f \in S(x(0))} \sup_{t \in [0,s]} |z_{k}^{-1} \mathbf{X}^{(k)}(z_kt) - x(t)| > \delta \right) = 0. \nonumber
\end{equation}
\label{thm:convergence}
\end{theorem}

This proof is standard and  is postponed later on in the appendix.

From the description of the dynamics, if we have $L(x(t)) = 0$, then $x_{i}(t) = 0$ for all $i$. Since $x(t)$ is a finite-dimensional vector, there exists at-least one coordinate $i^{*}(t)$  such that $x_{i^{*}(t)}(t) = L(x(t))$. Then one can write
\begin{align}
\frac{d}{dt}L(x(t)) &= \lambda \epsilon^{2} - \frac{C L(x(t))}{L \ln(2) \sum_{k=0}^{n_{\epsilon}-1} x_{k}(t) l_{\epsilon}(a_{i^{*}(t)},a_k)} \nonumber \\
& \leq \lambda \epsilon^2 - \frac{C}{L \ln(2) \sum_{k=0}^{n_{\epsilon}-1} l_{\epsilon}(a_k,0)},
\label{eqn:proof3_fluid_inequality}
\end{align}
where the second inequality comes by the fact that $x_{k}(t) \leq L(x(t))$ and the symmetry of the torus as given in Equation (\ref{eqn:proof3_torus_symmetry}). From Equation (\ref{eqn:proof3_fluid_inequality}), we see that under the condition given in (\ref{eqn:proof3_epsilon_stability_condn1}), $L(x(s)) = 0$ for all $s \geq \tau$ whenever $L(x(0)) = 1$, where $\tau = \left( \frac{C}{L \ln(2)  \sum_{k=0}^{n_{\epsilon}-1}l(a_k,0)} -  \lambda \epsilon^2  \right)^{-1}$, a deterministic time as defined in Equation (\ref{eqn:proof3_tau_defn}). 
\\

 We remark that inequality (\ref{eqn:proof3_fluid_inequality}) will be identical even in the case of arbitrary link distance $T$ and hence, the rest of the proof ingredients are the same for both when $T=0$ and $T > 0$.

\begin{lemma}
If condition in Equation (\ref{eqn:proof3_epsilon_stability_condn1}) holds, then
\begin{equation}
\lim_{L(x) \rightarrow \infty} \frac{1}{L(x)}\mathbb{E}_{x}[ |\mathbf{X}(L(x)\tau)|] = 0. 
\label{eqn:proof3_L1_Convergence}
\end{equation}
where $\tau$ is defined in Equation (\ref{eqn:proof3_tau_defn})
\label{lemma:fluid1}
\end{lemma}

\begin{proof}
The first observation is that the family of random variables $\left\{ \frac{|\mathbf{X}_{x}(|x|t)|}{|x|}\right\}_{x \in \mathbb{N}^{n_{\epsilon}} \setminus \{0\}}$ is uniformly integrable. Indeed, let $\{A_i(\cdot)\}_{i=0}^{n_{\epsilon}-1}$ be i.i.d. unit rate PPP denoting the arrivals into cell $i$. Then
\begin{align}
 X_{i}(t) \leq X_{i}(0) + A_{i}(\lambda \epsilon^2 t).
 \label{eqn:bounding_stability}
 \end{align}
 Thus, for $\mathbf{X}(0) = x$,
 \begin{align}
 \frac{X_{i}(|x|\tau)}{|x|} \leq \frac{x_{i}}{|x|} + \frac{A_{i}(\lambda \epsilon^2 |x| \tau)}{|x|}.
 \end{align}
 We have that $\frac{x_{i}}{|x|} \leq 1$ and the mean of $\frac{A_{i}(\lambda \epsilon^2 |x| \tau)}{|x|}$ equal to $\lambda \epsilon^2 \tau$. The variance of $\frac{A_{i}(\lambda \epsilon^2 |x| \tau)}{|x|}$ is $\frac{\lambda \epsilon^2 \tau}{|x|} \leq \lambda \epsilon^2 \tau$ for all $x \in \mathbb{N}^{n_{\epsilon}} \setminus \{0\}$. As the variance is uniformly bounded, the random variables $\left\{\frac{X_{i}(|x|\tau)}{|x|}\right\}_{x \in \mathbf{N}^{N_{\epsilon}} \setminus \{0 \}}$ are uniformly integrable. In addition, $\frac{|\mathbf{X}(|x|\tau)|}{|x|} \leq \sum_{i} \frac{X_{i}(|x| \tau)}{|x|}$, gives that  $\left\{ \frac{|\mathbf{X}_{x}(|x|\tau) |}{|x|}\right\}_{x \in \mathbf{N}^{N_{\epsilon}} \setminus \{0\}}$ is uniformly integrable since it is bounded above by a finite sum of random variables belonging to uniformly integrable families. 
 \\

 Let $x_k$ be any sequence of initial conditions such that $|x_k| \rightarrow \infty$. This implies that $a _k = {\mathbf{X}}^{(k)}(0)/|x_k| = x_k/|x_k|$ with $a_k \in [-1,1]^{n_{\epsilon}}$ for all $k$. Since the cube $[-1,1]^{n_{\epsilon}}$ is compact, there is a convergent sub-sequence i.e. $\frac{\mathbf{X}^{k(l)}(0)}{|x_{k(l)}|} \rightarrow x(0)$ with $|x(0)| = 1$. From Theorem \ref{thm:convergence},   there is a further sub-sequence of $k(l)$ such that $\frac{{\mathbf{X}}^{k^{'}(l)}(|x_{k^{'}(l)}|\tau)}{|x_{k^{'}(l)}|} \rightarrow x(\tau)$ almost surely where the function $x(\cdot) \in S(x(0))$. Under the  stability condition (\ref{eqn:proof3_epsilon_stability_condn1}), we have that for any fluid-limit function $x(\cdot) \in S(x(0))$, $x(\tau) = 0$ whenever $|x(0)| \leq 1$. This establishes that given any arbitrary sequence of initial conditions $x_k$ with $|x_k| \rightarrow \infty$, one can find a further sub-sequence $k^{'}(l)$ such that 
\begin{equation}
 \lim_{k^{'}(l) \rightarrow \infty} \frac{1}{|x_{k^{'}(l)}|}|\mathbf{X}^{k^{'}(l)}(|x_{k^{'}(l)}| \tau)| = 0, \text{   } a.s.
\end{equation}
Therefore, we can conclude that for any sequence $x_k$ with $|x_k| \rightarrow \infty$, we have $\frac{1}{|x_k|}|\mathbf{X}^{k}(|x_{k}|\tau)|$ tends to $0$  in probability.  But since, the family of random variables $\left\{ \frac{|\mathbf{X}_{x}(|x|\tau) |}{|x|}\right\}_{x \in \mathbf{N}^{N_{\epsilon}} \setminus \{0\}}$ is uniformly integrable, we have that 
 \begin{equation}
 \lim_{k \rightarrow \infty} \frac{1}{|x_k|}\mathbf{E}[|\mathbf{X}^{k}(|x_k| \tau)|] = 0.
 \end{equation}
 As $x_k$ was an arbitrary sequence, Equation (\ref{eqn:proof3_L1_Convergence}) holds whenever  condition (\ref{eqn:proof3_epsilon_stability_condn1}) holds.
\end{proof}

 From Equation (\ref{eqn:proof3_L1_Convergence}), we have that for any $\epsilon > 0$, there is a large enough $A_{\epsilon}$  such that Equation (\ref{eqn:proof3_markov_lyapunov_thm}) holds. Furthermore, for any finite $A$, the set $\{x \in \mathbf{N}^{N_{\epsilon}}:||x||_{\infty} \leq A\}$ is finite.  Hence, we have that $\mathbf{X}(t)$ is stable under the stability condition (\ref{eqn:proof3_epsilon_stability_condn1}) which proves Theorem \ref{thm:sufficient_condn}.
\\

\subsubsection*{Generalization to arbitrary Link Distance $T$}

To generalize the proof for arbitrary link distances $T$, we need to construct the appropriate discretization of the chain $\phi_{t}^{(\epsilon)}$. Once, we construct an appropriate discrete state space chain, then it is easy to see that the fluid version of this chain will satisfy inequality (\ref{eqn:proof3_fluid_inequality}) and Lemma \ref{lemma:fluid1} as is. This will conclude that the case with arbitrary link distance $T$ also yields the same stability result.
\\

The discrete state space process in this case will  naturally involve two vectors $\{X_i(t)\}_{i = 1}^{n \epsilon}$ and $\{Y_i(t)\}_{i = 1}^{n \epsilon}$, which represent the vector of transmitters and receivers in the discrete grid. However, in addition, we need a list of vectors $\{ \mathcal{M}_i (t)\}_{i = 1}^{n_{\epsilon}}$ where $\mathcal{M}_{i}$ is a $n_{\epsilon}$ dimensional vector whose  $j$th coordinate denotes how many   transmitters in cell $i$ have a corresponding receiver in cell $j$. The triple $\mathbf{X}_{t} := (X_i(t), Y_i(t), \mathcal{M}_i(t))_{i = 1}^{n_{\epsilon}}$ then evolves in a Markovian fashion on a countable state-space. The  evolution of $\mathbf{X}(t)$ is as follows. To each cell $i$, a new receiver is born at rate $\lambda \epsilon^2$. When, a receiver is born in cell $i$, we first pick an uniformly random location in the cell $A_i \subset \mathbf{S}$ and then centered around this point, we draw a ball of radius $T$ and pick the location of the transmitter uniformly on the circumference to decide the cell in which the transmitters land. Thus, at the instant of birth, both a transmitter and receiver is born. Thus, conditioned on the event that a receiver is placed in cell $i$, there is a distribution on the set $\{1,2,\cdots, n_{\epsilon} \}$ from which we sample the cell to place the corresponding transmitter in. To compute the interference seen at any receiver, we sum up the interference power from all transmitters in $\{Y_i\}_{i = 1}^{n_{\epsilon}}$ including the intended signaling transmitter, which forms an upper bound on the interference. On the event of a death of a receiver in cell $i$, we also delete an uniformly random transmitter such that it has a receiver in cell $i$. 
\\

This process $\mathbf{X}(t)$ can be studied using fluid limits as above but with significantly more computations. The fluid equations for this case (which is the analog of Theorem \ref{thm:convergence}) will be as follows.
\begin{align*}
\frac{d}{dt} x_i = \lambda \epsilon^2 - \frac{x_i l(T)}{ \sum_{j = 1}^{n^{\epsilon} } y_i l_{\epsilon}(a_i, a_j)  }
\end{align*}
whenever $x_i > 0$, else $\frac{d}{dt}x_i = 0$. Since, the number of transmitters and receivers are the same at all instants of time, we get the following inequality  immediately
\begin{align*}
\frac{d}{dt} y_i \leq \lambda \epsilon^2 - \frac{y_i l(T)}{ ||y||_{\infty} \sum_{j = 1}^{n^{\epsilon} } l_{\epsilon}(a_i, a_j)  }
\end{align*}
whenever $y_i > 0$. This is an inequality and not an equality  due to the fact that the interference is measured by a transmitter process $||y||_{\infty} \mathbf{1}$ which coordinate wise dominates the original transmitters $y$. Thus, we can see that by employing the  Lyapunov function $L(z)$ for $z = (x,y,\mathcal{M})$ as $L(z) := ||y||_{\infty}$, we will get exactly the same inequality as in Equation (\ref{eqn:proof3_fluid_inequality}). Furthermore, it is easy to check that Lemma \ref{lemma:fluid1} holds as is with $|\mathbf{X}(t)| := ||y||_{\infty}$. This will then establish that Theorem \ref{thm:XofT}
will hold as is for the chain constructed in this paragraph with generalized link distance $T$, which concludes the proof. 

\end{proof}

\section{Proof of Theorem \ref{thm:clustering}}
\label{proof_clustering}
\begin{proof}
The proof idea is to apply  Rate-Conservation equations similar to that of Theorem \ref{thm:necessary_condn}. For any receiver-transmitter pair $(x;y) \in \phi_t$, define $B_t(x) = \sum_{T \in \phi_{t}^{tx} \setminus \{ y\}}f(||T-x||)$ and the cadlag process $\mathcal{B}_t = \sum_{x \in \phi_{t}^{Rx}} B_t(x)$.

 Since we assume that the dynamics $\phi_t$ is ergodic, we write RCL for the stochastic process $\mathcal{B}_t$

\begin{align}
\lambda |S| \mathbb{E} \left[ 2 \int\limits_{x \in \mathbf{S}} B_0(x)\frac{dx}{|\mathbf{S}|}  \right] &= \lambda_d \mathbb{E}\left[ \frac{\mathbf{R}_0}{\mathbb{E}[\mathbf{R}_0]}  \sum_{T_n \in \phi_{0}} \frac{R(T_n, \phi_{0} )}{  \mathbf{R}_0 } 2 B_0(T_n) \right]
\label{eqn:RCL_arb_func1}
\end{align} 
The LHS follows from PASTA and the fact that a birth can happen anywhere in $\mathbf{S}$ uniformly and independently. The RHS follows from the Papangelou's theorem that the point process on $\mathbb{R}$ corresponding to the death epochs admits $\mathbf{R}_t = \frac{1}{L}\sum_{X_n \in \phi_0}R(X_n,\phi_0)$ as its Stochastic Intensity with respect to the  filtration $\mathcal{F}_t = \sigma \left( \{ \phi_s: s \leq t \}\right)$, the sigma algebra generated by the location of the links. We also have $\lambda_d = \lambda|S|$ from Equation (\ref{eqn:proof1_RCL_Number}) and $\mathbb{E}[\mathbf{R}_0] = \lambda |S| $ from Equation (\ref{eqn:proof1_RCL_bits1}). Using this to simplify Equation (\ref{eqn:RCL_arb_func1}), we get
\begin{align}
\mathbb{E}[B_0(0)]  &= \frac{1}{\lambda |S| L} \mathbb{E}\left[\sum_{T_n \in \phi_{0}} R(T_n, \phi_{0} )  B_0(T_n)\right],
\end{align}
where we used Fubini's theorem and the fact that $\phi_0$ is stationary in simplifying the LHS. Using the definition of Palm probability to simplify the RHS, we get
\begin{align}
\mathbb{E}[B_0(0)] &= \frac{\beta |\mathbf{S}|}{\lambda L |\mathbf{S}|} \mathbb{E}^{0}_{\phi_0} \left[ R(0,\phi_0) B_0(0)\right].
\label{eqn:RCL_arb_func_2}
\end{align}
Since both $f(\cdot)$ and the path-loss $l(\cdot)$ are positive non-increasing functions, we have the deterministic behavior that if $B_0(0)$ increases, then $R(0,\phi_0)$ decreases. Hence, we can use the association inequality
\begin{align}
\mathbb{E}^{0}_{\phi_0} \left[ R(0,\phi_0) B_0(0)\right] \leq \mathbb{E}^{0}_{\phi_0} \left[ R(0,\phi_0) \right]  \mathbb{E}^{0}_{\phi_0} \left[B_0(0)\right]
\label{eqn:RCL_arb_func_3}
\end{align}

Employing Inequality (\ref{eqn:RCL_arb_func_3}) in Equation (\ref{eqn:RCL_arb_func_2}), and the RCL $\lambda L = \beta \mathbb{E}^{0}_{\phi_0}[R(0,\phi_0)]$ from equation (\ref{eqn:proof1_RCL_bits2}) we get
\begin{align}
\mathbb{E}[B_0(0)] \leq \mathbb{E}^{0}_{\phi_0}[B_0(0)].
\end{align}
\end{proof}

\section{Proof of Theorem \ref{thm:convergence}}

\begin{proof} 
This can be argued by contradiction. Assume that for some $\epsilon>0$ and a sub-sequence 
\begin{equation}
\mathbb{P} \left( \inf_{f \in S(x(0)} \sup_{t \in [0,T]} |z_{k}^{-1}X(z_kt) - f(t)| > \epsilon \right) \geq \epsilon
\end{equation}
Without loss of generality,  assume  the above holds true for all $k \geq 1$. 
\\

The trajectories of the process $X^{k}(t)$ can be written in terms of independent unit-rate Poisson process $A_{i}^{k}$ and $D_{i}^{k}$ 
\begin{multline}
X_{i}^{k}(t) = X_{i}^{k}(0) + A_{i}^{k}(\lambda \epsilon^2 t) -  D_{i}^{k}\left( \int_{0}^{t} X_{i}^{k}(u) \log_2 \left( 1 + \frac{1}{N_0 + I_{i}^{\epsilon}(t)} du \right)  \right).
\label{eqn:functional_trajectory}
\end{multline} 

That is, $X^{k}(t)$ is a functional of the Point Process satisfying the set of Equations (\ref{eqn:functional_trajectory}).
\\

One can rewrite equation (\ref{eqn:functional_trajectory})  by a change of variables as 
\begin{multline}
\frac{1}{z_k}X^{k}_{i}(z_kt) = \frac{1}{z_k}X^{k}_{i}(0) + \frac{1}{z_k} A^{k}_{i}(\lambda \epsilon^2 z_kt) -  \frac{1}{z_k} D_{i}^{k}\left( \int_{0}^{z_kt} X_{i}^{k}(u) \log_2 \left( 1 + \frac{1}{N_0 + I_{i}^{\epsilon}(t)} \right) du  \right).
\end{multline}

Now replacing $u$ by $z_kl$, we get the following
\begin{multline}
\frac{1}{z_k}X^{k}_{i}(z_kt) = \frac{1}{z_k}X^{k}_{i}(0) + \frac{1}{z_k} A^{k}_{i}(\lambda \epsilon^2 z_kt) -  \frac{1}{z_k} D_{i}^{k}\left( z_k \int_{0}^{t} X_{i}^{k}(z_kl) \log_2 \left( 1 + \frac{1}{N_0 + I_{i}^{\epsilon}(z_kl)} \right) dl  \right),
\end{multline}
which can be written as 
\begin{multline}
\frac{1}{z_k}X^{k}_{i}(z_kt) = \frac{1}{z_k}X^{k}_{i}(0) + \lambda \epsilon^2 t -    \int_{0}^{t} X_{i}^{k}(z_kl) \log_2 \left( 1 + \frac{1}{N_0 + I_{i}^{\epsilon}(z_kl)} \right) dl + \delta_{i}^{k}(t),
\label{eqn:fluid_limit_error}
\end{multline}

where the error term $\delta_{i}^{k}(t)$  satisfies the stochastic bound 
\begin{multline}
\sup_{t \in [0,T]}|\delta_{i}^{k}(t)| \leq \frac{1}{z_k} \sup_{t \in [0,\lambda \epsilon^2 T]} |A_{i}^{k}(z_kt) - z_kt|  +   \frac{1}{z_k} \sup_{t \in [0,T \log_2(e)]} |D_{i}^{k}(z_kt) - z_kt|.
\end{multline}
 The error term $\delta_{i}^{k}(t)$ is bounded by the following lemma.
\begin{lemma} \cite{massoulie2007}
Let $\Xi$ be a unit rate PPP on the real line. Then for all $T > 0$ and all $\lambda > 0$,
\begin{equation}
\mathbb{P} \left( \sup_{t \in [0,T]} | \Xi(t) - t| \geq \lambda T \right) \leq e^{-Th(\lambda)} + e^{-Th(-\lambda)}
\end{equation}
where $h(\lambda) = (1 + \lambda) \log (1 + \lambda) - \lambda$.
\end{lemma}

This lemma in particular implies that there exists a sub-sequence $k(l),l \geq 1$ and a sequence $\epsilon(l) \rightarrow0$ such that $\forall i$
\begin{align*}
\sum_{l \geq 1} \mathbb{P} \left( \sup_{t \in [0,T]}| \delta_{i}^{k(l)}(t)| \geq \epsilon(l) \right) < \infty
\end{align*}
By Borel-Cantelli's lemma, there exists a sub sequence  such that for all $i$, $\lim_{l \rightarrow \infty}\sup_{t \in [0,T]} |\delta_{i}^{k(l)}(t)| \rightarrow 0$  almost surely. 
\\

Now consider the random function $w_k(t) = \int_{0}^{t} X_{ij}^{k}(z_kl) \log_2 \left( 1 + \frac{1}{N_0 + I_{i}^{\epsilon}(z_kl)} \right) dl$ which is Lipschitz for each sample path $\omega$, i.e.

\begin{align}
w_k(t) - w_k(s) &= \int_{s}^{t} X_{ij}^{k}(z_kl) \log_2 \left( 1 + \frac{1}{N_0 + I_{i}^{\epsilon}(z_kl)} \right) dl \\
& \leq (t-s)\frac{\log_2(e)}{\sup_{x,y \in \mathbf{S}} l^{\epsilon}(x,y)} < \infty.
\end{align} 

 From the Arzela-Ascoli theorem,  there exists a sub-sequence such that $w_k(t)$ converges uniformly on $[0,T]$ to a Lipschitz continuous function $D_{i}(t)$ for each sample path $\omega$. This along with the bound on $\sup_{t \in [0,T]}|\delta_{ij}(t)|$ yields that there is a sub-sequence such that
\begin{align}
\frac{X_{i}^{k}(z_kt)}{z_k} \rightarrow x_{ij}(t) := x_{i}(0) + \lambda_{i}t - D_{i}(t) , \text{  } a.s.,
\end{align}
where the convergence happens uniformly over $[0,T]$. $D_{i}(t)$ is Lipschitz since $x_{i}(t)$ is Lipschitz continuous. It remains to show that $\dot{D}_{i}(t) = \frac{x_{i}(t)}{I_{i}(t)}$. Since $D_{i}(t)$ is Lipschitz continuous, by Rademacher's theorem, it is differentiable almost everywhere on $[0,T]$. For all $h>0$,
\begin{equation}
\int_{t}^{t+h} X_{i}^{k}(z_kl) \log_2 \left( 1 + \frac{1}{N_0 + I_{i}^{\epsilon}(z_kl)} \right) dl \rightarrow \int_{t}^{t+h} \frac{x_{i}(l)}{I_{i}^{\epsilon,f}(l)}dl. \nonumber
\end{equation}
This follows from dominated convergence and the Lipschitz continuity of  $l \rightarrow x_{ij}(l)$. Therefore $\dot{D}_{i}(t) = \frac{x_{i}(t)}{I_{i}^{\epsilon,f}(t)}$.

Hence, we have shown that given any sequence of initial conditions $X^{k}(0)$ and number $z_k$ such that the limit $\frac{X^{k}(0)}{z_k} = x(0)$ exists, we can find a sub-sequence $k_l$ such that $\frac{X^{k_l}(z_{k_l}t)}{z_{k_l}}$ converges almost surely to the Lipschitz continuous fluid limit function $x(t)$. This is a contradiction and hence the theorem is proved.
\end{proof}

\section{Proof Sketch for Corollary \ref{cor:mimo_result}}
\label{sec:mimo_results}
We just provide the proof outline for the necessary condition. The sufficient condition follows identically as in proof of Theorem \ref{thm:sufficient_condn}. For the necessary condition, note that all Rate-Conservation Equations  (Equations \ref{eqn:proof1_RCL_Number} , \ref{eqn:proof1_RCL_bits1} and \ref{eqn:proof1_RCL_interference4}) hold. In particular, we only have a different upper bound for $R(x;\phi_0)I(x;\phi_0)$ since the rate-function used is a different one. From Equation (\ref{eqn:mimo_rate_defn_indep}), we have 
\begin{align}
\lambda L \int_{x \in \mathbf{S}}l(|x|)dx &\leq \lim_{q \rightarrow \infty} \mathbb{E}_{H} \left[  \sum_{i=1}^{X_r}  q\log_2 \left(  1 + \frac{1}{X_t (N_0 + 1q)}  \sigma_i (H H^{\dag})  \right)    \right] \nonumber \\
& \stackrel{(a)}{=} \mathbb{E}_{H} \left[  \sum_{i=1}^{X_r}  \lim_{q \rightarrow \infty} q  \log_2 \left(  1 + \frac{1}{X_t (N_0 + q)}  \sigma_i (H H^{\dag})  \right)     \right] \nonumber \\
&= \frac{\log_2(e)}{X_t} \mathbb{E}_{H} \left[\sum_{i=1}^{X_r} \sigma_i(H H^{\dag}) \right] \nonumber \\
& \stackrel{(b)}{=} \log_2(e) X_r ,\label{eqn:mimo_proof_penultimate}
\end{align}
where $(a)$ follows from the Monotone Convergence theorem and $(b)$ follows from the fact that $H$ is a matrix whose entries are i.i.d. complex-normal random variables with zero mean and unit variance. Re-arranging inequality (\ref{eqn:mimo_proof_penultimate}) yields the necessary condition on the stability region for the MIMO channel model with independent channels.


\end{document}